\documentclass[journal12pt]{IEEEtran}
\usepackage{graphicx}
\usepackage{mathptmx}
\usepackage{amsmath}
\usepackage{amssymb}
\usepackage{amsthm}

\newtheorem{prop}{Proposition}

\theoremstyle{definition}

\usepackage{kantlipsum} 
\usepackage{mwe} 
\usepackage{caption,subcaption}
\usepackage{amssymb}
\usepackage{graphicx}
\usepackage{bm}
\usepackage{commath}
\usepackage{mathtools}
\usepackage{amsmath}
\usepackage{mathtools, nccmath}
\usepackage{mathrsfs}
\usepackage{amsfonts}
\usepackage{siunitx}
\usepackage{textcomp}

\newtheorem{remark}{Remark}

\usepackage{makecell}
\usepackage{boldline}
\usepackage{array,booktabs,arydshln,xcolor}
\usepackage{cuted, nccmath}
\DeclareMathOperator{\atan}{atan}
\usepackage{lipsum}
\usepackage{stfloats}
\let\OLDthebibliography\thebibliography
\renewcommand\thebibliography[1]{
  \OLDthebibliography{#1}
  \setlength{\parskip}{0pt}
  \setlength{\itemsep}{0pt plus 0.3ex}
}

\ifCLASSINFOpdf
\else
\fi
\begin{document}
%
\title{A Swash Mass Unmanned Aerial Vehicle:\\ Design, Modeling and Control}
%
%
%
\author{Andrea~M. Tonello
        and~Babak~Salamat
\thanks{Andrea M. Tonello is with the Alpen-Adria-Universit\"{a}t Klagenfurt, Institute of Networked and Embedded Systems, 9020, Austria, e-mail: andrea.tonello@aau.at.}
\thanks{Babak Salamat is with the Alpen-Adria-Universit\"{a}t Klagenfurt, Institute of Networked and Embedded Systems, 9020, Austria, e-mail: babaksa@edu.aau.at.}}

\maketitle

\begin{abstract}
In this paper, a new unmanned aerial vehicle (UAV) structure, referred to as swash mass UAV, is presented. It consists of a double blade coaxial shaft rotor and four swash masses that allow changing the orientation and maneuvering the UAV. The dynamical system model is derived from the Newton\textquotesingle s law framework. The rotational behavior of the UAV is discussed as a function of the design parameters. Given the uniqueness and the form of the obtained non-linear dynamical system model, a back-stepping control mechanism is proposed. It is obtained following the Lyapunov's control approach in each iteration step. Numerical results show that the swashed mass UAV can be maneuvered  with the proposed control algorithm so that linear and aggressive trajectories can be accurately tracked.
\end{abstract}

\begin{IEEEkeywords}
Aerospace, Unmanned aerial vehicle, dynamical system model, non-linear control, back-stepping control, trajectory tracking.
\end{IEEEkeywords}

%
\IEEEpeerreviewmaketitle

\section{Introduction}
%
%
%
%

\IEEEPARstart{O}{ver} the past few decades unmanned aerial vehicles (UAVs) have received growing attention for their versatility in different application domains. Initially conceived for military applications, nowadays they find deployment in surveillance systems, aerial photography, traffic control, and agriculture. In addition, their navigation autonomy, compact size, low environmental impact when electrically propelled, drive the development of new commercial business models for logistics, e.g., parcel delivery, and future urban transportation services, e.g., aerial taxi.
To fulfill the diverse requirements of the cited applications, the design of new mechanical structures is a quite vivid research activity. Such structures, essentially, can be divided into three significant categories: fixed-wing (FW) UAVs, rotary-wing (RE) UAVs and hybrid-layout (HL) UAVs \cite{austin2010unmanned}. The fixed-wing UAVs have fixed-wings appropriately shaped and positioned to produce lift from the forward movement of the vehicle. FW UAVs reach high-speeds and have the ability to fly over long distances. However, they require horizontal take-off and landing procedures. On the other hand, rotary-wing UAVs offer vertical take-off and landing capabilities as well as the ability to hover in a motionless spot. Quadrotor helicopters are a popular example of RW UAVs. The rotor disks are aligned in a single plate and the quadrotor's thrust vector is constrained vertically to the plate. Indeed, rotary-wing UAVs are not aerodynamically optimized as fixed-wing UAVs and reach lower speeds and flight distance. To increase their limited maneuverability \cite{7487497,6504725}, a number of modified configurations have been developed. For instance, the authors in \cite{6907516, Long2013, 6868215} propose a multi-rotor helicopter with the ability to tilt the propellers so that the thrust vector direction can be changed. To combine the advantages of both FW and RW UAVs, hybrid-layout UAVs have been conceived and they deploy both wings and thrust rotors. However, they are mechanically complex due to the rotor inclination apparatus \cite{Garcia:2010:MCM:1965425}.
In all UAV systems, a key component is the automatic control mechanism. Different control methodologies have been studied for the trajectory tracking problem. A linear quadratic regulator (LQR) approach has been considered  in \cite{8396552, 7784809, 8395008} for trajectory tracking of a quadrotor UAV in the presence of external disturbances. Bounded tracking control was analyzed in \cite{6213080} for a wide class of reference trajectories. Stochastic feedback control was proposed for a fixed-wing UAV in \cite{6780633, 8469075}. Interconnection and damping assignment passivity based control (IDA-PBC) has been applied in \cite{1556727, 1024334} for the vertical takeoff and landing of a degree one under actuated aircraft with strong input coupling. Finally, since the dynamical system model is often nonlinear and with a number of coupled subsystems, non-linear back-stepping control mechanisms have been considered valuable. For instance, Lyapunov based back-stepping control has been investigated for position tracking of a quadrotor UAV in \cite{4287130, 8390703}.

\begin{figure}[t]
\includegraphics[scale=0.8]{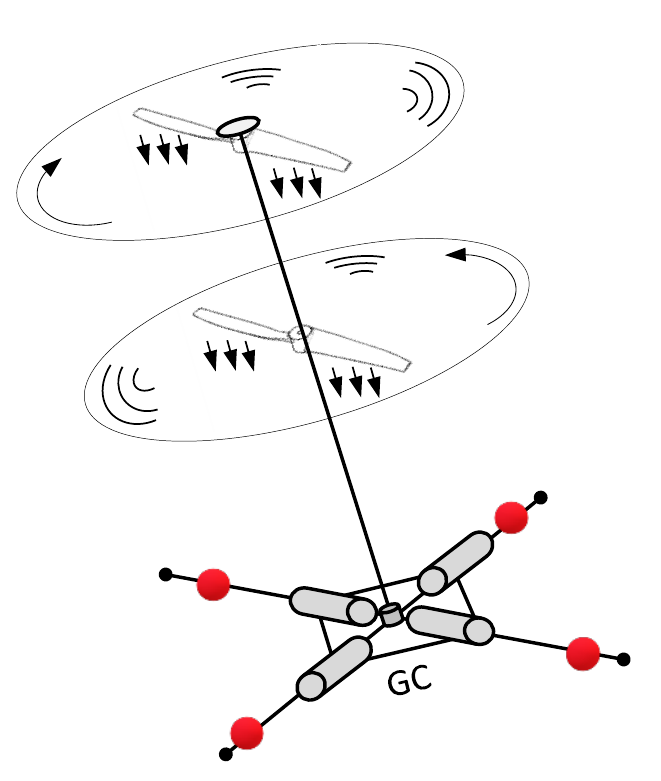}
\centering
\vspace{-0.1em}
\caption{Swash mass unmanned aerial vehicle structure. }
\vspace{-1.0em}
\label{Fig_UAV}
\end{figure}

In this paper, the main contribution is twofold. We first propose and model a novel unmanned aerial vehicle structure (Fig.~\ref{Fig_UAV}).  Then, we address the control mechanism for trajectory tracking.  In reference to the UAV structure, the basic idea is

\begin{center}
\begin{figure*}[h]
\includegraphics[width=6.4in,height=4.4in,clip,keepaspectratio]{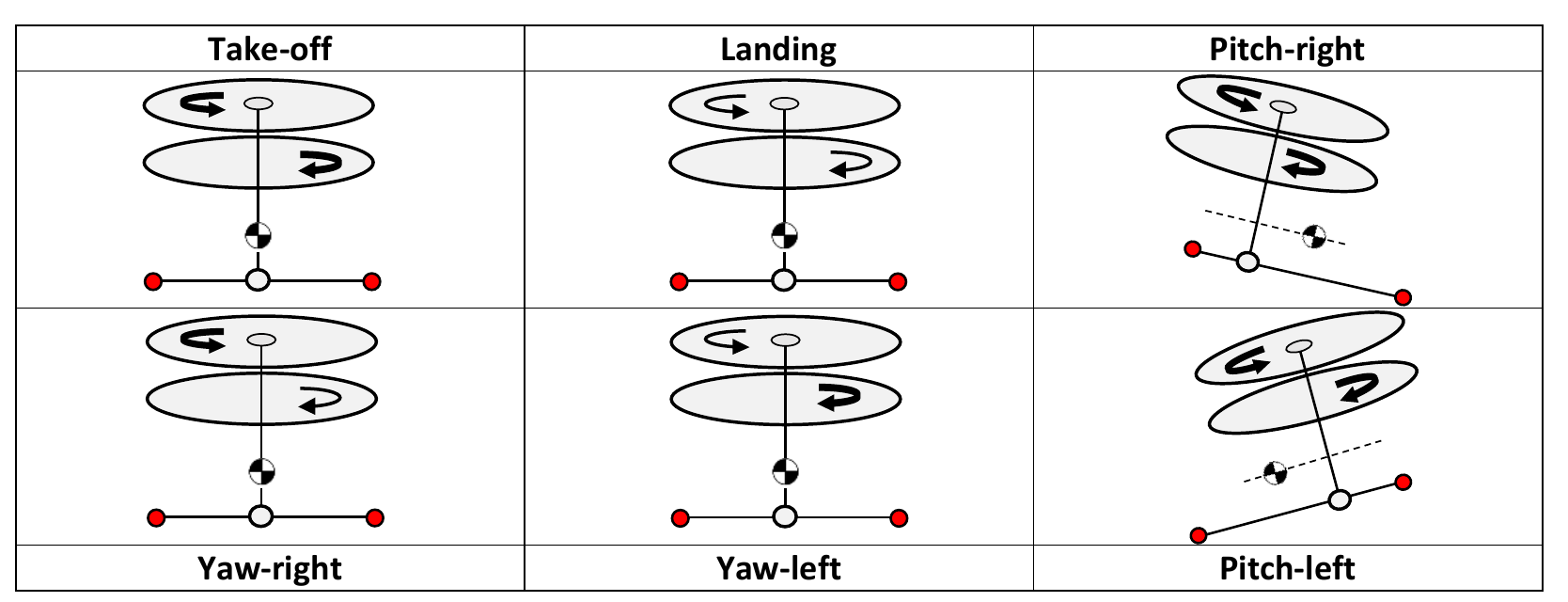}
\centering
\vspace{-0.1em}
\caption{Take-off, landing and yaw motion of the swash mass UAV.  }
\vspace{-1.0em}
\label{Fig_Basic}
\centering
\end{figure*}
\end{center}

to deploy a double blade coaxial shaft rotor and to maneuver the helicopter through the linear movement of four masses positioned on the main body plate. We refer to it as the swashed mass helicopter. The rotors provide a drag thrust, while the inertial masses (through the gravitational forces) induce a certain orientation so that to attain a certain roll, pitch and yaw. By controlling the rotors speed and the swash masses displacement it is possible (as it will be shown) to control the UAV to follow certain trajectories, as well as to provide VTOL and hovering ability. In contrast, traditional helicopters, for instance those with body and tail rotor wings, deploy a collective pitch swash-plate \cite{watkinson2003art} that is capable to change the main blade pitch angle, and then they adapt the main and tail blade speeds to induce certain maneuvers. However, the swash plate mechanisms is mechanically complex and of difficult realization in small UAVs.

In more detail, the specific contributions of this paper are:
\begin{itemize}
\item the description of the swash mass helicopter structure;
\item the derivation of the dynamical system model;
\item the effect of sizing on control input response;
\item the development of a control strategy using a non linear back-stepping control methodology (\cite{zhou2008adaptive}, Chapter 2) to both stabilize the rotor craft and track a certain trajectory.
\end{itemize}

The derived back-stepping control mechanism nicely applies to the proposed novel UAV structure since the dynamic system model can be divided into a fully actuated subsystem and a coupled underacted subsystem. Furthermore, the dynamical system model includes the derivative of the control inputs, as a result of a time variant inertia matrix, which makes the control design challenging. To overcome this challenge and derive the control law, we propose a simplified state space dynamic model by approximating the inertia matrix with a constant value.

The reminder of this paper is organized as follows. The basic characteristics of the UAV structure are given in Section II. The mathematical formulation of the system dynamics is presented in Section III. The equilibrium analysis as well as design guidelines (sizing) of the structure are given in Section IV. The control problem is addressed in Section V. Several simulation results are presented in Section VI. The conclusions then follow. \par
\textbf{Notation:} Vectors and matrices are denoted with bold letters. Unless stated otherwise, all vectors in this paper are column vectors. The vector cross product is denoted with $\times$. The state variables are a function of time, e.g., $\bm{x} = \bm{x}(t)$. The first and second derivative with respect to time of a state space variable $x$ are denoted respectively with $\dot{x}$, and $\ddot{x}$.

\begin{center}
\begin{figure*}[h]
\includegraphics[width=6.5in,height=4in,clip,keepaspectratio]{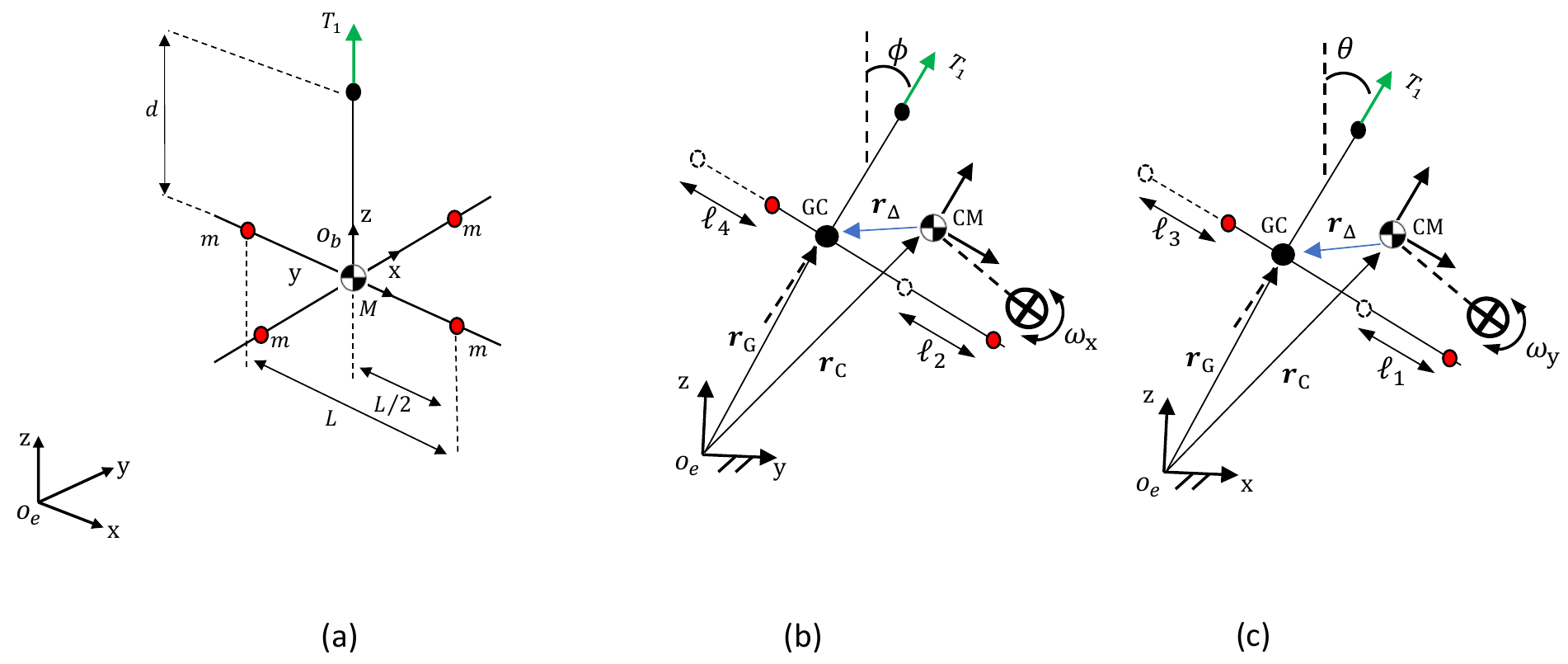}
\centering
\vspace{-0.1em}
\caption{Free diagram of the swash mass UAV. (a) Reference frames $O_{e}$ and $O_{b}$. The connection between the frames is denoted with radius vectors $\bm{r}_{G}$, $\bm{r}_{c}$ and $\bm{r}_{\Delta}$. (b) The planner UAV equipped with two swash masses in $y-z$ plane. Main body (GC), swash masses $2m$ and the center of mass (CM). Showing only pitch dynamics for clarity. (c) The planner UAV equipped with two swash masses in $x-z$ plane.}
\vspace{-1.0em}
\label{Frames}
\centering
\end{figure*}
\end{center}

\section{Basic elements of the swash mass uav}
The proposed UAV deploys a coaxial double blade rotor with the rotor connected to the main helicopter body via a rigid shaft. To steer the helicopter, four masses are displaced on an orthogonal plane w.r.t. to the rotor shaft and can moved with linear cross shaft servos. Assuming the rotor and blades to be concentrated in one point in the rotor shaft edge, we refer to such a point as the rotor center (RC). The intersection of the rotor shaft and the swash masses plane is referred to as geometrical center (GC) of the UAV. The blades rotation induces a thrust aligned with the rotor shaft and the swash masses shifts the center of mass (CM) of the UAV on an orthogonal plane to the thrust vector, i.e.,. it tilts the UAV body (Fig.~\ref{Fig_Basic}). A yaw movement is generated by changing the relative speed of the two blades since a drag torque imbalance is generated, while roll and pitch are regulated by displacing the swash masses asymmetrically w.r.t. the GC.

\section{Dynamical system model}
In this section, firstly, the coordinate systems will be defined. Then, the full non-linear dynamical system model describing the UAV's translation and rotation will be derived using the Newton's framework.

\subsection{Coordinate Systems}
To describe the behaviour (translation and rotation) of the UAV, three coordinate systems can be defined: the inertial reference frame, the center of mass (CM) frame aligned with the inertial reference frame and attached to the CM and the body-fixed reference frame. Note that we express angular momentum and angular velocities w.r.t. the CM frame.

The body-fixed reference frame $O_{b} = \pmb{\{} x_{b}, y_{b}, z_{b} \pmb{\}}$ is centered in the GC (see Fig.~\ref{Frames}). The $z_{b}$ axis has the orientation of the rotor shaft and it is orthogonal to the swash masses plane. The two mass cross shafts are aligned with the $x_{b}$ and $y_{b}$ axes.

The inertial reference frame is $O_{e} = \pmb{\{} x, y, z \pmb{\}}$ and it is fixed on the earth ground with the gravity vector pointing towards the negative $z$ direction\footnote{We assume that the Earth-fixed frame can be considered as an inertial reference frame since the effect of Earth's rotation on an aerial object moving from the North pole towards the equator can be neglected \cite{etkin2012dynamics}.}.

We also use the following notation for radius and velocity vectors: $\bm{r}_{C}$ denotes radius vector from the inertial reference frame $O_{e}$ to the CM frame while $(\bm{r}_{\Delta})_{b}$ denotes the radius vector from the CM frame to the body-fixed reference frame $O_{b}$ expressed in the body-fixed reference frame. $\bm{r}_{G}$ denotes the radius vector from the inertial reference frame $O_{e}$ to the GC of the UAV.

Since the body reference frame is fixed to the unmanned aerial vehicle structure, it translates and rotates with the body \cite{dynamic}. Its rotation w.r.t. the CM frame is given by the Euler angles $(\phi, \theta, \psi)$.

Finally, we can relate the body-fixed reference frame to the inertial reference frame by the rotation matrix $\bm{R}^{I}_{B}$ which passes a vector from the former frame to the latter \cite{etkin2012dynamics}. Its rotation w.r.t. the inertial reference frame (shifted so that it has the same origin of the body reference frame) is given by the rotation matrix
\begin{equation}
\label{Rotation}
\bm{R}^{I}_{B} = \begin{bmatrix}
                   c_{\theta} c_{\psi} & s_{\phi}s_{\theta}c_{\psi}-c_{\phi}s_{\psi} & c_{\phi}s_{\theta}c_{\psi} + s_{\phi}s_{\psi} \\[2pt]
                   c_{\theta} s_{\psi} & s_{\phi}s_{\theta}s_{\psi}-c_{\phi}c_{\psi} & c_{\phi}s_{\theta}s_{\psi} - s_{\phi}c_{\psi} \\[2pt]
                   -s_{\theta} & s_{\phi}c_{\theta} & c_{\phi}c_{\theta}
                 \end{bmatrix}.
\end{equation}
The terms $s.$ and $c.$ represent the sine and cosine functions of the argument in the subscript, respectively.
\subsection{External Forces and Moments}
We assume that the UAV has a mass $M$, including the four swash masses each with identical value $m$, which is concentrated in the CM of the UAV. The UAV is subjected to translational and rotational forces in a constant gravitational field with the gravity $\bm{g}$ that is aligned to the inertial reference frame z-axis. Therefore, the overall gravitational force expressed in the inertial reference frame can be written as follows
\begin{equation}
\label{gravity_b}
\bm{F}_{g,I} = \begin{bmatrix}
                         0 \\
                         0 \\
                         -Mg
                       \end{bmatrix}.
\end{equation}
The thrust is the main force generated by the rotation of the two blades. Overall, the net thrust $T_{1}$ (which is aligned with the rotor shaft) can be written in the shifted inertial reference frame (with origin RC) as follows
\begin{equation}
\label{Force}
\bm{F}_{r,I} = \bm{R}^{I}_{B} \bm{F}_{r,B} = \bm{R}^{I}_{B} \begin{bmatrix}
            0 \\
            0 \\
            T_{1}
          \end{bmatrix},
\end{equation}
where $T_{1}$ (which is directly proportional to the upper and lower rotor aerodynamic coefficient $\gamma_{1}$ and to the rotational speeds $\Omega_{1}$ and $\Omega_{2}$ \cite{watkinson2003art}) equals $T_{1} = \gamma_{1}(\Omega^{2}_{1} + \Omega^{2}_{2})$.

The UAV is tilted by steering the swash masses since they generate a moment vector about an axis passing through the CM. Such a moment vector can be expressed in the shifted inertial frame (CM frame) as follows
\begin{align}\nonumber
\begin{split}
\bm{M}_{C} = \bm{R}^{I}_{B}(\bm{r}_{\Delta})_{b} \times \bm{R}^{I}_{B}\begin{bmatrix}
                                   0 \\
                                   0 \\
                                   T_{1}
                                 \end{bmatrix} + \bm{R}^{I}_{B}\begin{bmatrix}
                 0 \\[2pt]
                0 \\[2pt]
                 \gamma_{2}(\Omega^{2}_{1} - \Omega^{2}_{2})
               \end{bmatrix}
               \end{split}\\
               \begin{split}\label{Mc}
               = \bm{R}^{I}_{B}\bigg(     (\bm{r}_{\Delta})_{b} \times \begin{bmatrix}
                                   0 \\
                                   0 \\
                                   T_{1}
                                 \end{bmatrix}     \bigg) + \bm{R}^{I}_{B}\begin{bmatrix}
                 0 \\[2pt]
                0 \\[2pt]
                 \gamma_{2}(\Omega^{2}_{1} - \Omega^{2}_{2})
               \end{bmatrix}.
\end{split}
\end{align}
The yaw moment $ M_{\psi} = \gamma_{2}(\Omega^{2}_{1} - \Omega^{2}_{2})$ generated by the rotating blades is directly proportional to the corresponding aerodynamic coefficient $\gamma_{2}$ and rotation speeds difference \cite{watkinson2003art}.

\subsection{Complete Non-linear Dynamical System Model}
The translational dynamics of the CM w.r.t. the inertial reference frame can be obtained by applying the first cardinal Newton's equation \cite{etkin2012dynamics, Prof_Tonello} and exploiting (\ref{gravity_b}) and (\ref{Force}). The model becomes:
\begin{equation}
\label{Translational_e}
     \begin{bmatrix}
        \ddot{x} \\
        \ddot{y} \\
        \ddot{z}
      \end{bmatrix}_{C} = \dfrac{1}{M} \bm{F}_{g,I} + \dfrac{1}{M}\bm{F}_{r,I}.
\end{equation}
Our objective is to derive the translational dynamics of the geometric center (GC) of the helicopter observed in the inertial reference frame. To do so, we consider a shift (see Fig.~\ref{Frames}\hspace{0.2em}(b-c)), and write
\begin{equation}\label{shift}
\bm{r}_{G} = \bm{r}_{C} + \bm{r}_{\Delta}.
\end{equation}
We use the fact that $\bm{r}_{\Delta} = \bm{R}^{I}_{B}(\bm{r}_{\Delta})_{b}$, where $\bm{R}^{I}_{B}$ and $(\bm{r}_{\Delta})_{b}$ are the rotation matrix from the body-fixed frame to the CM frame and the radius vector from the CM frame to the body-fixed reference frame expressed in the body-fixed frame, respectively. Taking the time derivative of both sides of (\ref{shift}) and using the identity $\bm{\dot{R}}^{I}_{B}(\bm{r}_{\Delta})_{b} = \bm{R}^{I}_{B}(\bm{\omega} \times (\bm{\dot{r}}_{\Delta})_{b})$ \cite{goldstein2002classical}, we get
\begin{equation}\label{dshift}
\bm{\dot{r}}_{\Delta} = \bm{R}^{I}_{B}\bigg[  (\bm{\dot{r}}_{\Delta})_{b}  + \bm{\omega} \times (\bm{r}_{\Delta})_{b}    \bigg],
\end{equation}
where  $\bm{\omega} = [\omega_{x} \hspace{0.4em} \omega_{y}\hspace{0.4em}  \omega_{z}]^T$ is the angular velocity vector. Taking one more time derivative we obtain
\begin{equation}\label{ddshift}
\bm{\ddot{r}}_{\Delta} = \bm{R}^{I}_{B}\bigg[  (\bm{\ddot{r}}_{\Delta})_{b}  +   2\bm{\omega} \times (\bm{\dot{r}}_{\Delta})_{b} + \bm{\dot{\omega}} \times (\bm{r}_{\Delta})_{b} + \bm{\omega} \times \big( \bm{\omega} \times  (\bm{r}_{\Delta})_{b} \big)          \bigg].
\end{equation}
Substituting (\ref{ddshift}) in (\ref{shift}), we obtain the translational dynamics of the GC w.r.t. the inertial reference frame. To make our notation simple, we define
\begin{equation}\label{Simple notation}
\begin{bmatrix}
  x \\
  y \\
  z
\end{bmatrix}_{GC} \triangleq \begin{bmatrix}
                                x \\
                                y \\
                                z
                              \end{bmatrix}.
\end{equation}
Therefore, we obtain
\begin{align}\nonumber
\begin{split}
 \begin{bmatrix}
        \ddot{x} \\
        \ddot{y} \\
        \ddot{z}
      \end{bmatrix} = \bm{R}^{I}_{B}\bigg[      (\bm{\ddot{r}}_{\Delta})_{b}  +   2\bm{\omega} \times (\bm{\dot{r}}_{\Delta})_{b} +  \bm{\dot{\omega}} \times (\bm{r}_{\Delta})_{b}
\end{split}\\
\begin{split}\label{GC}
+ \bm{\omega} \times \big( \bm{\omega} \times  (\bm{r}_{\Delta})_{b} \big)          \bigg] + \dfrac{1}{M} \bm{F}_{g,I} + \dfrac{1}{M}\bm{F}_{r,I},
\end{split}
\end{align}
where $(\bm{r}_{\Delta})_{b}$ can be easily computed as follows:
\begin{align}\nonumber
\begin{split}
(\bm{r}_{\Delta})_{b} = - \dfrac{\sum_{i=1}^{4}m_{i}\bm{r}_{b,i}}{m_{b} + \sum_{i=1}^{4}m_{i}} = -\dfrac{\sum_{i=1}^{4}m_{i}\bm{r}_{b,i}}{M}
\end{split}\\
\begin{split}\nonumber
= -\dfrac{m}{M}\bigg(\begin{bmatrix}
                        \frac{L}{2} + \ell_{1} \\
                        0 \\
                        0
                      \end{bmatrix}   + \begin{bmatrix}
                        -\frac{L}{2} + \ell_{3} \\
                        0 \\
                        0
                      \end{bmatrix}   + \begin{bmatrix}
                        0\\
                        \frac{L}{2} + \ell_{2}  \\
                        0
                      \end{bmatrix}
\end{split}\\
\begin{split}\label{rdelta}
  + \begin{bmatrix}
                        0\\
                        -\frac{L}{2} + \ell_{4}  \\
                        0
                      \end{bmatrix}      \bigg) = -\beta\begin{bmatrix}
                                                                  \ell_{1} + \ell_{3} \\
                                                                  \ell_{2} + \ell_{4} \\
                                                                  0
                                                                \end{bmatrix},
\end{split}
\end{align}
where we define $\beta$ as the ratio of the swash mass weight and the total UAV weight $\beta=\frac{m}{M}$, $m_{b}$ denotes the mass of the UAV rigid body (without steering masses), $\bm{r}_{b,i}$ represents the position of the  \textit{i}-th swash masses, expressed in the body-fixed frame \footnote{Note that we assume that the origin of the GC coincides with the CM of the UAV during hovering.}. $\frac{L}{2}$ and $\ell_{i}$ denote the rest position of the \textit{i}-th swash masses during hovering and the instantaneous position of the \textit{i}-th swash mass, respectively.

We now determine the angular momentum of our UAV about the center of mass CM, which may have an acceleration $\bm{a}_{c} = [\ddot{x}, \hspace{0.2em} \ddot{y}, \hspace{0.2em} \ddot{z}]_{c}^{T}$.
The rotational motion can be obtained by applying the second cardinal Newton's equation (Euler's moment equation ) \cite{Prof_Tonello} as follows
\begin{equation}\label{Rotational Motion1}
\bigg(\dfrac{d\bm{L}}{dt}\bigg)_{C} = \bm{M}_{C}.
\end{equation}
The angular momentum $\bm{L}$ is equal to
\begin{equation}\label{L}
\bm{L} = \sum_{i=1}^{5}\bm{r}_{i} \times \bm{Q}_{i} = \sum_{i=1}^{5}\bm{r}_{i} \times m_{i}\bm{\upsilon}_{i} =  \sum_{i=1}^{5}\bm{r}_{i} \times (\bm{\omega} \times \bm{r}_{i})m_{i} = \bm{I}\bm{\omega},
\end{equation}
where $\bm{r}_{i}$ is the position vector relative to the CM of the representative particle of mass $m_{i}$, $m_{i} = m$ for $i=1, ..., 4$ and $m_{5} = m_{b}$. For our system, the velocity of $m_{i}$ relative to CM is $\dot{\bm{r}}_{i} = \bm{\omega} \times \bm{r}_{i}$. $\bm{I} \in \mathbb{R}^{3 \times 3}$ is the overall inertia tensor, or inertia matrix of the system. In computing the angular momentum $\bm{L}$ it should be taken into account the fact that, due to the swash masses, the inertia matrix is time-dependent. Furthermore, we assume the two swash masses to be mutually constrained at constant distance $L$ and $\frac{L}{2} + \ell_{2} - (-\frac{L}{2} +\ell_{4}) = L$ and $\frac{L}{2} + \ell_{1} - (-\frac{L}{2} +\ell_{3}) = L$. Since, $\ell_{2} = \ell_{4}$ and $\ell_{1} = \ell_{3}$, for simplicity we define $\ell_{y} = 2\ell_{2}$ and $\ell_{x} = 2\ell_{1}$.
The elements of the inertia matrix can be computed as follows

\begin{align}\nonumber
\begin{split}
I_{xx} =  m_{b}\bigg[  -2\beta\ell_{y}  \bigg]^{2}
\end{split}\\
\begin{split}\label{Ixx}
+ m\bigg[ \bigg(\dfrac{1}{2}-2\beta\bigg)\ell_{y} + \dfrac{L}{2}   \bigg]^{2} + m\bigg[ \bigg(\dfrac{1}{2}-2\beta\bigg)\ell_{y} - \dfrac{L}{2}          \bigg]^{2}
\end{split}\\
\begin{split}\nonumber
I_{yy} =  m_{b}\bigg[  -2\beta\ell_{x}  \bigg]^{2}
\end{split}\\
\begin{split}\label{Iyy}
+ m\bigg[ \bigg(\dfrac{1}{2}-2\beta\bigg)\ell_{x} + \dfrac{L}{2}  \bigg]^{2} + m\bigg[ \bigg(\dfrac{1}{2}-2\beta\bigg)\ell_{x} - \dfrac{L}{2}          \bigg]^{2}
\end{split}
\end{align}
\begin{align}
\begin{split}\nonumber
I_{zz} =  m_{b}\bigg[  -2\beta\ell_{y}  \bigg]^{2} + m_{b}\bigg[  -2\beta\ell_{x}  \bigg]^{2}
\end{split}\\
\begin{split}\nonumber
 + m\bigg[ \bigg(\dfrac{1}{2}-2\beta\bigg)\ell_{x} + \dfrac{L}{2}   \bigg]^{2} + m\bigg[ \bigg(\dfrac{1}{2}-2\beta\bigg)\ell_{x} - \dfrac{L}{2}          \bigg]^{2}
\end{split}\\
\begin{split}\nonumber
 + m\bigg[ \bigg(\dfrac{1}{2}-2\beta\bigg)\ell_{y} + \dfrac{L}{2} \bigg]^{2} + m\bigg[ \bigg(\dfrac{1}{2}-2\beta\bigg)\ell_{y} - \dfrac{L}{2}          \bigg]^{2}
\end{split}
\end{align}
\begin{align}
\begin{split}\nonumber
I_{xy} = I_{yx} = -m_{b}\Bigg[\bigg[  -2\beta\ell_{y}\cos(\phi)  \bigg]\bigg[  -2\beta\ell_{x}\cos(\theta)  \bigg]\Bigg]
\end{split}\\
\begin{split}\nonumber
-m\bigg[ \big[ \bigg(\dfrac{1}{2}-2\beta\bigg)\ell_{x}\cos(\theta)   + \dfrac{L}{2}\cos(\theta)  \big]
\end{split}\\
\begin{split}\nonumber
 \big[ \bigg(\dfrac{1}{2}-2\beta\bigg)\ell_{y}\cos(\phi) + \dfrac{L}{2}\cos(\phi) \big]        \bigg]
\end{split}\\
\begin{split}\label{Ixy}
-m\bigg[ \big[ \bigg(\dfrac{1}{2}-2\beta\bigg)\ell_{x}\cos(\theta)   - \dfrac{L}{2}\cos(\theta)  \big]
\end{split}\\
\begin{split}\nonumber
 \big[ \bigg(\dfrac{1}{2}-2\beta\bigg)\ell_{y}\cos(\phi) - \dfrac{L}{2}\cos(\phi) \big]        \bigg]
\end{split}\\
\begin{split}\label{Ixzy}
I_{xz} = I_{yz} = 0.
\end{split}
\end{align}
We are in the position to write the full dynamical system model.

\textit{Full Swash Mass UAV Dynamical System Model:}
From the results above, the full non-linear dynamical system model of the UAV can be written as
\begin{align}\nonumber
\begin{split}
 \begin{bmatrix}
        \ddot{x} \\
        \ddot{y} \\
        \ddot{z}
      \end{bmatrix} = \bm{R}^{I}_{B}\bigg[      (\bm{\ddot{r}}_{\Delta})_{b}  +   2\bm{\omega} \times (\bm{\dot{r}}_{\Delta})_{b} +  \bm{\dot{\omega}} \times (\bm{r}_{\Delta})_{b}
\end{split}\\
\begin{split}\label{Full_3D_model T}
+ \bm{\omega} \times \big( \bm{\omega} \times  (\bm{r}_{\Delta})_{b} \big)          \bigg] + \dfrac{1}{M} \bm{F}_{g,I} + \dfrac{1}{M}\bm{F}_{r,I},
\end{split}
\end{align}
\begin{equation}\label{Full_3D_model R}
\bm{\dot{I}}\bm{\omega} + \bm{I}\bm{\dot{\omega}} = \bm{M}_{C}.
\end{equation}

\begin{remark}
The model has six outputs $\{x, \hspace{0.2em} y, \hspace{0.2em}z, \hspace{0.2em}\phi, \hspace{0.2em}\theta, \hspace{0.2em}\psi \}$ and only four control inputs $\{T_{1}, \hspace{0.2em}\ell_{y}, \hspace{0.2em}\ell_{x}, \hspace{0.2em}M_{\psi}\}$. Therefore, the UAV is an under-actuated dynamical system. Furthermore, it should be noted that such a dynamical system is not common since the inertia matrix depends on the control inputs which in turn make it time dependent. This is taken into account in the relation (\ref{Full_3D_model R}).
\end{remark}

\begin{center}
\begin{figure*}[h]
\includegraphics[width=4.6in,height=3in,clip,keepaspectratio]{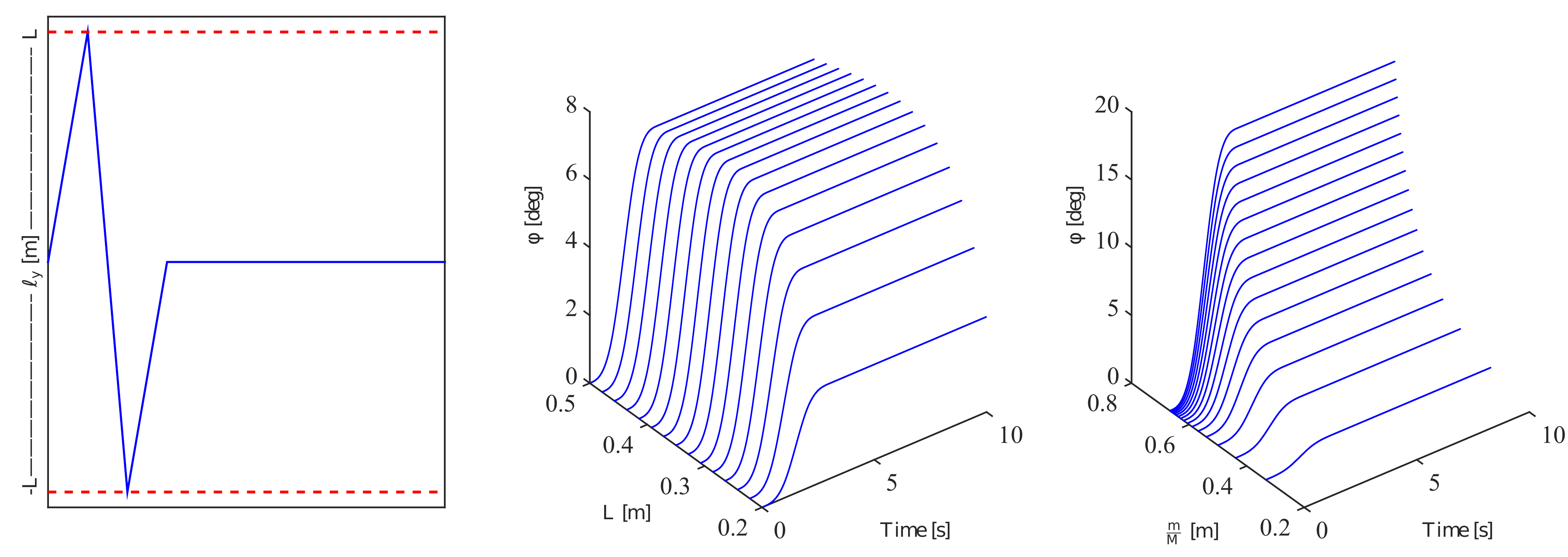}
\centering
\vspace{-0.1em}
\caption{Effect of sizing on control input response. a) a saw triangle function $\ell_{y}$. b) Relation between the pitch angle $\phi$ and the arm length $L$ for a given constant mass ratio $\beta=\frac{m}{M} = 0.09$. c) Relation between the pitch angle $\phi$ and the mass ratio $\beta = \frac{m}{M}$ for a given constant $L = 0.3 m$. }
\vspace{-1.0em}
\label{Sizing}
\centering
\end{figure*}
\end{center}

\begin{figure*}[h]
\includegraphics[scale=0.23]{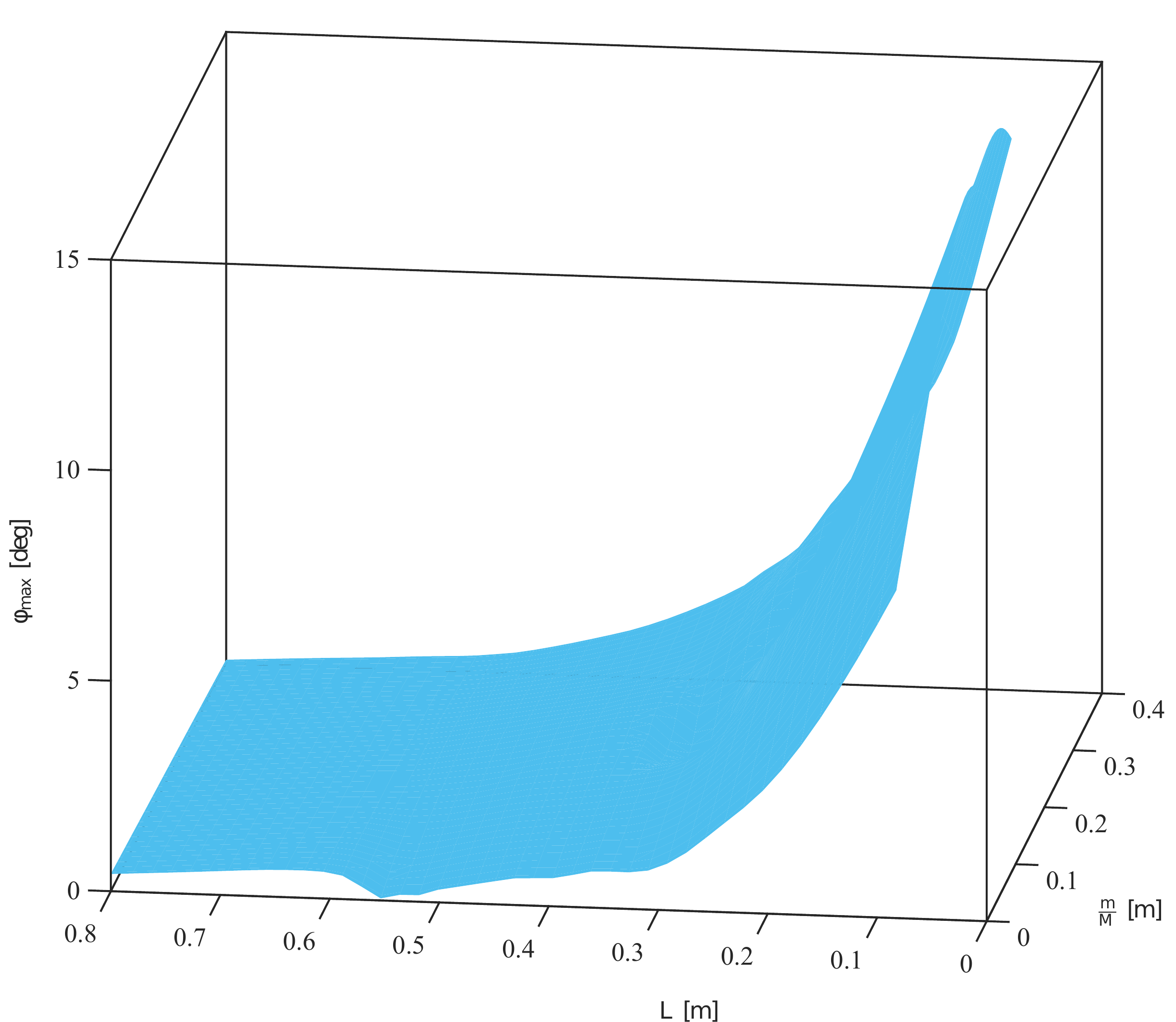}
\centering
\vspace{-0.1em}
\caption{Maximum pitch angle $\phi_{max}$ as a function of the mass ratio $\beta$ and arm length $L$. }
\vspace{-1.0em}
\label{Sizet}
\end{figure*}

\section{effect of sizing on control input response}
In this section, we want to get an understanding of the motion behavior of the considered UAV. This helps the dimensioning/sizing of its structure as a function of the design parameters, and to introduce the control problem (discussed in the next section). We start by analyzing the effect of sizing on the control input response.

For the sake of exposition simplicity, to analyze the effect of sizing on control input response, we assume a 2D planner scenario so that the UAV structure can be depicted as in (Fig.~\ref{Frames}\hspace{0.2em}(b)). The planner dynamics, obtained from (\ref{Full_3D_model T}) and (\ref{Full_3D_model R}), then becomes as follows

\textit{2D Dynamical System Model}: If we consider a two dimensions UAV structure, the dynamical system model reads as follows
\begin{align}
\begin{split}\nonumber
M\begin{bmatrix}
   \ddot{y} \\[2pt]
   \ddot{z}
\end{bmatrix}  = \beta\begin{bmatrix}
                     2\dot{\phi}\dot{\ell}_{y}\sin(\phi) - \ddot{\ell}_{y}\cos(\phi) + \ell_{y}\ddot{\phi}\sin(\phi) + \ell_{y}\dot{\phi}^{2}\cos(\phi) \\[2pt]
                     -\ddot{\ell}_{y}\sin(\phi) + \ell_{y}\dot{\phi}^{2}\sin(\phi) - 2\dot{\phi}\dot{\ell}_{y}\cos(\phi) - \ell_{y}\ddot{\phi}\cos(\phi)
                  \end{bmatrix}
\end{split}\\
\begin{split}\label{Full_2D_model T}
                  + \begin{bmatrix}
                         T_{1}\sin(\phi) \\[2pt]
                         T_{1}\cos(\phi)
                       \end{bmatrix} - \begin{bmatrix}
                                         0 \\[2pt]
                                         Mg
                                       \end{bmatrix}
\end{split}\\
\begin{split}\nonumber
\bigg[ m_{b}(  -2\beta\ell_{y}  )^{2} + m\big((\dfrac{1}{2}-2\beta)\ell_{y} + \dfrac{L}{2} \big)^{2}
\end{split}\\
\begin{split}\nonumber
+ m\big((\dfrac{1}{2}-2\beta)\ell_{y} + \dfrac{L}{2}\big)^{2}    \bigg]\ddot{\phi}
\end{split}\\
\begin{split}\label{Full_2D_model R}
+\bigg[ \ell_{y}\dot{\ell}_{y}(m - 8\beta m + 16\beta{^2}m + 8\beta^{2}m_{b}) \bigg]\dot{\phi} = \beta T_{1}\cos(\phi)\ell_{y}.
\end{split}
\end{align}
In this case, the control inputs are limited to $\bm{u}= [T_{1}, \hspace{0.2em}\ell_{y}]^{T}$. While the angle $\phi$ in (\ref{Full_2D_model R}) is determined only by the UAV main design parameters $M$, $m$, $L$, the swash mass position $\ell_{y}$ and the thrust $T_{1}$, the translational motion in (\ref{Full_2D_model T}) depends on the thrust $T_{1}$, swash mass position $\ell_{y}$, velocity $\dot{\ell}_{y}$, acceleration $\ddot{\ell}_{y}$ and the pitch angle $\phi$. Since, only the rotational dynamics in (\ref{Full_2D_model R}) is of importance at this point, we set the thrust around the hovering point, which is $T_{1} = \frac{Mg}{\cos(\phi)}$. We analyse the rotational dynamics in (\ref{Full_2D_model R}) for an elementary control input as a function of design parameters $M, m, L$. We consider the control input $\ell_{y}$ to be equal to a saw triangle function as shown in Fig.~\ref{Sizing}\hspace{0.2em}(a). In Fig.~\ref{Sizing},\hspace{0.2em} we plot the relation between:

a) the control input $\ell_{y}$; \par
b) the pitch angle response and the arm length $L$ for a given constant mass ratio $\beta = \frac{m}{M}$; \par
c) the pitch angle response and the mass ratio $\beta = \frac{m}{M}$ for a given constant $L$. \par
Fig.~\ref{Sizing}\hspace{0.2em}(b) shows that it is sufficient to have the arm length $L$ equal $0.4 m$ for a given constant mass ratio $\beta = \frac{m}{M} = 0.09$ to induce a high value for the pitch angle response. The pitch response increases sharply for a given constant $L = 0.3 m$ (Fig.~\ref{Sizing}\hspace{0.2em}(c)) as the mass ratio $\beta = \frac{m}{M}$ increases.
From the results, it can be seen that the pitch angle $\phi$ reaches a maximum value $\phi_{max}$. Therefore, in Fig.~\ref{Sizet} we plot the maximum value of pitch angle $\phi_{max}$ as a function of the pair made of the arm length $L$ and the mass ratio $\beta = \frac{m}{M}$. Overall, it can be seen that there exist several design solutions, i.e., choices of the parameters $M, m, L$, to have the UAV reach a given pitch angle.

The presented results suggest that to obtain better dynamical performance, one should increase $\beta = \frac{m}{M}$ indefinitely. In practice, the swash mass weight is limited due to the physical constraint of the UAV's structure and we can not distribute all the weight on the swash masses only. Therefore, the design strategy is to provide a trade-off between the design parameters of the UAV and the mission flight requirements.

An example of design parameters that will be used also in the numerical results section, is reported in Table~\ref{TableI}.

In the next section, we develop a control law that allows the UAV to follow a certain trajectory.

\section{trajectory tracking and control}
\label{Trajectory Tracking}
The main objective of the automatic control is to act on the rotor thrust and swash masses position so that the UAV can track a desired target trajectory $\{x_{\ast}, y_{\ast}, z_{\ast}\}$ with stable Euler angles. To proceed, for the sake of exposition and understanding simplicity, we focus first on the 2D planner. The derivation of the control laws in the full 3D case, will be given next (Section V). \par
In the 2D planner, we can make two observations: firstly, the dynamical system comprises two main sub-systems (altitude and horizontal coordinate relations in (\ref{Full_2D_model T}) that are coupled via the control input $T_1$ and the pitch angle $\phi$), and a third sub-system (pitch angle relation (\ref{Full_2D_model R}) with control input $\ell_{y}$ coupled to the other subsystems via the pitch angle $\phi$ itself). The pitch angle is determined by the control input $\ell_{y}$ only. Secondly, the dynamics of the pitch angle depends not only on the input $\ell_{y}$ but also on its derivative.

The first observation suggests the use of a non-linear back-stepping control mechanism. Back-stepping control is an iterative approach that breaks down the controller design into steps and therefore the control of each sub-system is implemented iteratively \cite{zhou2008adaptive, 4058900}. In detail, since we want to control the UAV via $T_{1}$ and $\ell_{y}$ to follow a certain target trajectory ($y_{\ast}(t), z_{\ast}(t)$), we first determine the control input $T_{1}$ to reach the target altitude. Then, we define the virtual control input $u_y = \sin(\phi)$ and derive a control law for it to reach the target horizontal coordinate. Finally, we regulate the value of $\ell_{y}$ so that the pitch angle $\phi$ derived from the virtual control input $u_y$ is obtained. Furthermore, each control law in the back-stepping process is derived via the Lyapunov methodology \cite{Khalil:1173048}.

The second observation suggests to proceed by deriving a dynamical system model where the inertia term in ({\ref{Full_2D_model R}) is approximated with an appropriate constant value. In particular, since the inertia term varies as the swash masses move from a minimum to a maximum displacement position, we propose to approximate it as
\begin{align}\nonumber
\begin{split}
I_{xx} \thickapprox I_{c} =  m_{b}\bigg(  -2\beta\ell_{mean}  \bigg)^{2}  +  m\bigg[ \bigg(\dfrac{1}{2}-2\beta\bigg)\ell_{mean} + \dfrac{L}{2}  \bigg]^{2}
\end{split}\\
\begin{split}\label{constant inertia}
 + m\bigg[ \bigg(\dfrac{1}{2}-2\beta\bigg)\ell_{mean} - \dfrac{L}{2}   \bigg]^{2},
 \end{split}
\end{align}
with
\begin{equation}\label{r_mean}
\ell_{mean} =  \bigg(\dfrac{\ell_{max} + \ell_{min}}{2}\bigg).
\end{equation}
and
\begin{align}\label{max and min variation}
\ell_{max} = L,  \hspace{0.6em} \ell_{min} = -L.
\end{align}

Therefore, the simplified dynamical system in (\ref{Full_2D_model R}) becomes as follows.\\

\textit{2D Simplified Dynamical System Model}: Under the assumption in (\ref{constant inertia}), the 2D dynamical system model reads as follows

\begin{align}
\label{Newton's_Law (2D) final}
\begin{split}
M\ddot{y}   =   \beta f_{1}(\phi, \dot{\phi}, \ell_{y}, \dot{\ell}_{y}, \ddot{\ell}_{y}) + T_{1}\sin(\phi) \\
\end{split}\\
\begin{split}\label{Newton's_Law (2D) finale1}
M\ddot{z}   = \beta f_{2}(\phi, \dot{\phi}, \ell_{y},\dot{\ell}_{y},\ddot{\ell}_{y}) -Mg + T_{1}\cos(\phi) \\
\end{split}\\
\begin{split}\label{Newton's_Law (2D) finale2}
I_{c}\ddot{\phi} =  \beta T_{1}\cos(\phi)\ell_{y}\\
\end{split}\\
\begin{split}\label{Newton's_Law (2D) finale3}
s.t.     \hspace{1em} - L \leqslant \ell_{y} \leqslant L.
\end{split}
\end{align}
where
\begin{align}
\begin{split}\nonumber
f_{1}(\phi, \dot{\phi}, \ell_{y},\dot{\ell}_{y},\ddot{\ell}_{y}) = 2\dot{\phi}\dot{\ell}_{y}\sin(\phi) - \ddot{\ell}_{y}\cos(\phi)
\end{split}\\
\begin{split}\label{F1}
+ \ell_{y}\ddot{\phi}\sin(\phi) + \ell_{y}\dot{\phi}^{2}\cos(\phi)
\end{split}
\end{align}
and
\begin{align}\nonumber
\begin{split}
f_{2}(\phi, \dot{\phi}, \ell_{y},\dot{\ell}_{y},\ddot{\ell}_{y}) =  -\ddot{\ell}_{y}\sin(\phi) + \ell_{y}\dot{\phi}^{2}\sin(\phi)
\end{split}\\
\begin{split}\label{F2}
- 2\dot{\phi}\dot{\ell}_{y}\cos(\phi) - \ell_{y}\ddot{\phi}\cos(\phi).
\end{split}
\end{align}

Before applying the backstepping approach, we need to be sure that two functions $f_{1}$ and $f_{2}$ in (\ref{Newton's_Law (2D) final}) and (\ref{Newton's_Law (2D) finale1}) are bounded.

\begin{prop}
Let $f_{1}(\phi, \dot{\phi}, \ell_{y},\dot{\ell}_{y},\ddot{\ell}_{y}) = 2\dot{\phi}\dot{\ell}_{y}\sin(\phi) - \ddot{\ell}_{y}\cos(\phi) + \ell_{y}\ddot{\phi}\sin(\phi) + \ell_{y}\dot{\phi}^{2}\cos(\phi)$ be a function defined over a set $\chi$, where all the variables $\phi, \dot{\phi}, \ell_{y},\dot{\ell}_{y},\ddot{\ell}_{y}$ in the set $\chi$ are bounded. Then, there exists a positive constant $\Theta_{1}$ such that
\begin{align}\nonumber
\begin{split}
f_{1}(\phi, \dot{\phi}, \ell_{y},\dot{\ell}_{y},\ddot{\ell}_{y}) \leq \sqrt{\ddot{\ell}^2_{y} + 4\dot{\phi}^2\dot{\ell}^{2}_{y}} + \sqrt{\ell^{2}_{y}\dot{\phi}^{4} + \ell^{2}_{y}\ddot{\phi}^{2}} \leq \Theta_{1}
\end{split}\\
\begin{split} \label{Bounded}
f_{1}(\phi, \dot{\phi}, \ell_{y},\dot{\ell}_{y},\ddot{\ell}_{y}) \geq -\sqrt{\ddot{\ell}^2_{y} + \dot{\phi}^2\dot{\ell}^{2}_{y}} -\sqrt{\ell^{2}_{y}\dot{\phi}^{4} + \ell^{2}_{y}\ddot{\phi}^{2}} \geq -\Theta_{1}
\end{split}
\end{align}
\end{prop}

\begin{proof}
$f_{1}(\phi, \dot{\phi}, \ell_{y},\dot{\ell}_{y},\ddot{\ell}_{y}) = 2\dot{\phi}\dot{\ell}_{y}\sin(\phi) - \ddot{\ell}_{y}\cos(\phi) + \ell_{y}\ddot{\phi}\sin(\phi) + \ell_{y}\dot{\phi}^{2}\cos(\phi) =  a\cos(\phi) + b\sin(\phi) + c\cos(\phi) + d\sin(\phi)$ where $a=-\ddot{\ell}_{y}$, $b=2\dot{\phi}\dot{\ell}_{y}$, $c= \ell_{y}\dot{\phi}^{2}$ and $d=\ell_{y}\ddot{\phi}$. Exploiting the trigonometric property, $f_{1}(\phi, \dot{\phi}, \ell_{y},\dot{\ell}_{y},\ddot{\ell}_{y})$ can be written as:
\begin{align}\nonumber
\begin{split}
f_{1}(\phi, \dot{\phi}, \ell_{y},\dot{\ell}_{y},\ddot{\ell}_{y}) =
\end{split}\\
\begin{split}\nonumber
a\cos(\phi) + b\sin(\phi) + c\cos(\phi) + d\sin(\phi)\equiv
\end{split}\\
\begin{split}\label{Trog}
\sqrt{a^2 + b^2}\cos\bigg(\phi-\atan(\dfrac{a}{b})\bigg)  +\sqrt{c^2 + d^2}\cos\bigg(\phi-\atan(\dfrac{c}{d})\bigg).
\end{split}
\end{align}
Therefore, we can write
\begin{align}
\begin{split}\nonumber
f_{1}(\phi, \dot{\phi}, \ell_{y},\dot{\ell}_{y},\ddot{\ell}_{y}) \leq \sqrt{\ddot{\ell}^2_{y} + 4\dot{\phi}^2\dot{\ell}^{2}_{y}} + \sqrt{\ell^{2}_{y}\dot{\phi}^{4} + \ell^{2}_{y}\ddot{\phi}^{2}} \leq \Theta_{1}
\end{split}\\
\begin{split}\label{Bounded in proof}
= \sqrt{\ddot{\ell}^2_{y,max} + 4\dot{\phi}_{max}^2\dot{\ell}^{2}_{y,max}} + \sqrt{\ell^{2}_{y,max}\dot{\phi}_{max}^{4} + \ell^{2}_{y,max}\ddot{\phi}_{max}^{2}}.
\end{split}
\end{align}
The last inequality in (\ref{Bounded in proof}) comes from the fact that $a, b, c$ and $d$ are assumed bounded. By symmetry, we can conclude that $-\sqrt{\ddot{\ell}^2_{y} + 4\dot{\phi}^2\dot{\ell}^{2}_{y}} - \sqrt{\ell^{2}_{y}\dot{\phi}^{4} + \ell^{2}_{y}\ddot{\phi}^{2}} \geq -\Theta_{1}$ is also bounded.
\end{proof}
The same bounding methodology can be applied to the second function and we obtain $ f_{2}(\phi, \dot{\phi}, \ell_{y},\dot{\ell}_{y},\ddot{\ell}_{y}) \leq \sqrt{\ddot{\ell}^2_{y} + 4\dot{\phi}^2\dot{\ell}^{2}_{y}} + \sqrt{\ell^{2}_{y}\dot{\phi}^{4} + \ell^{2}_{y}\ddot{\phi}^{2}} \leq \Theta_{2}$.

Now, let $\bm{u} = [T_{1}, \hspace{0.5em}\ell_{y}]^T$ be the control input vector. The dynamics equations (\ref{Newton's_Law (2D) final})-(\ref{Newton's_Law (2D) finale3}) can be written in a state-space realization form $f(\bm{x}, \bm{u})$ by defining $\bm{x} = [x_{1}, ..., x_{6}]^T$ as the state vector:
\begin{equation}\label{State-space}
\begin{bmatrix}
x_{1} \\[2pt]
x_{2}
\end{bmatrix}= \begin{bmatrix}
                  y \\[2pt]
                  \dot{y}
                \end{bmatrix}, \begin{bmatrix}
x_{3} \\[2pt]
x_{4}
 \end{bmatrix}= \begin{bmatrix}
                  z \\[2pt]
                  \dot{z}
                \end{bmatrix}, \begin{bmatrix}
x_{5} \\[2pt]
x_{6}
\end{bmatrix}= \begin{bmatrix}
                  \phi \\[2pt]
                  \dot{\phi}
                \end{bmatrix}.
 \end{equation}
Consequently, the dynamical system model of the UAV using the state variables in (\ref{State-space}) can be written as
\begin{align}\label{state-space-dynamic}
\begin{split}
\dot{x}_{1} = x_{2} \\
\end{split}\\
\begin{split} \label{state-space-dynamic1}
\dot{x_{2}} = \dfrac{\beta \Theta_{1} + T_{1}\sin(x_{5})}{M} \\
\end{split}\\
\begin{split} \label{state-space-dynamic2}
\dot{x_{3}} = x_{4} \\
\end{split}\\
\begin{split} \label{state-space-dynamic3}
\dot{x_{4}} = -g + \dfrac{\beta \Theta_{2} + T_{1}\cos(x_{5})}{M} \\
\end{split}\\
\begin{split} \label{state-space-dynamic4}
\dot{x_{5}} = x_{6} \\
\end{split}\\
\begin{split}\label{state-space-dynamic5}
\dot{x_{6}} =  \dfrac{\beta T_{1} \cos(x_{5})\ell_{y}}{I_{c}}\\
\end{split}
\end{align}

\begin{table}[h]
\renewcommand{\arraystretch}{1.3}
\caption{Example of design parameters of the UAV}
\label{TableI}
\centering
\begin{tabular}{c c c c}
\bfseries Parameters & \bfseries Symbol & \bfseries Value & \bfseries Unit   \\\hlineB{2.5}
\hline \hline
Mass               & $M$ & $1.1$ & $kg$\\
Swash mass               & $m$ & $0.1$ & $kg$\\
Maximum displacement of the mass             & $L$ & $0.2$ & $m$\\
Gravitational acceleration     & $g$      & $9.81$ & $m/s^{2}$\\
\\\bottomrule
\end{tabular}
\end{table}

\begin{figure*}[h]
\includegraphics[scale=0.5]{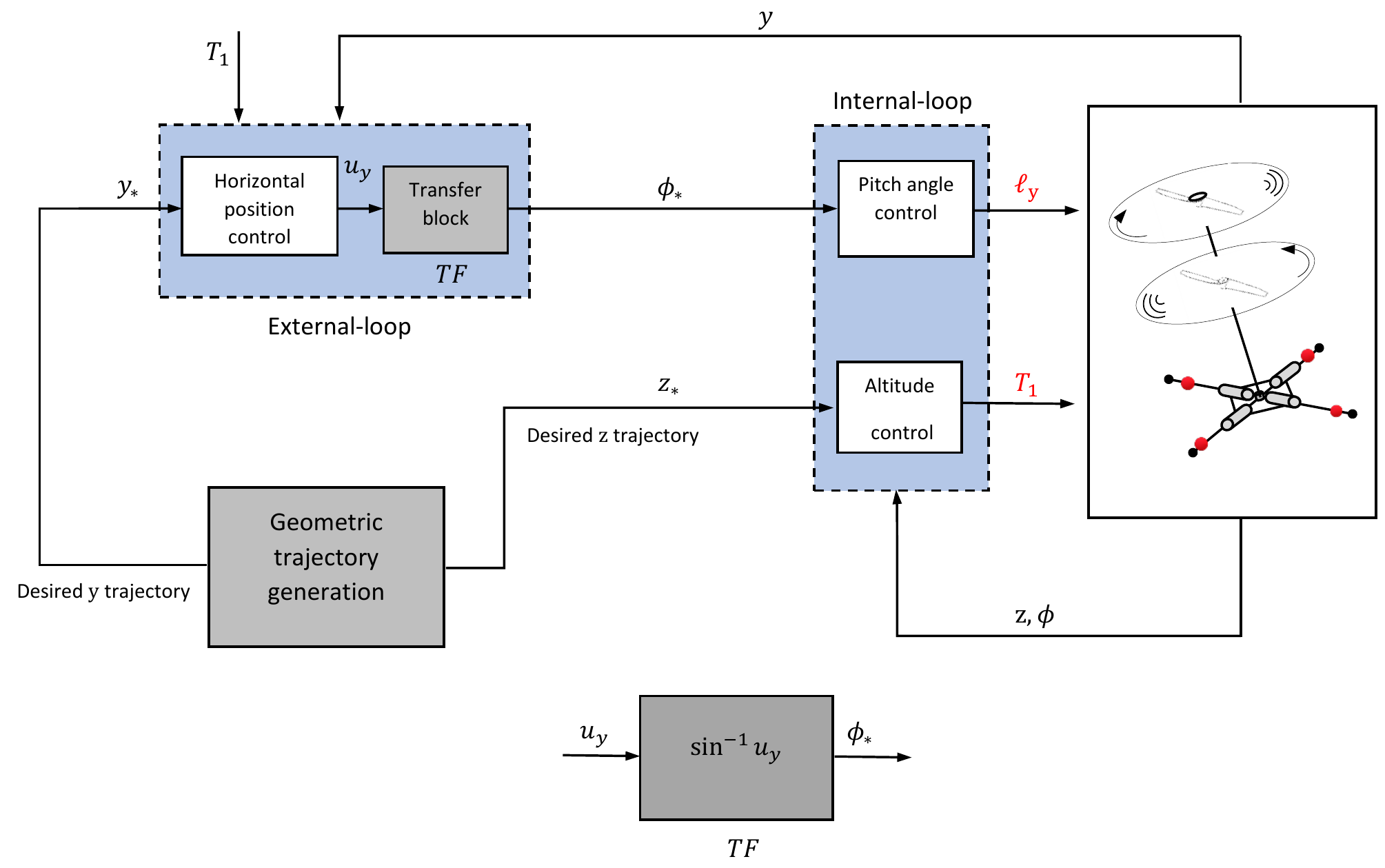}
\centering
\vspace{-0.1em}
\caption{Back-stepping controller with internal and external control loop architecture, with focus to the 2D case. }
\vspace{-1.0em}
\label{Architecture}
\end{figure*}

\subsection{Altitude and Horizontal Position Control}
The control mechanism for the proposed UAV starts by formulating a control law for $T_{1}$ (the total thrust generated by two rotor blades), so that given any initial pitch angle $\phi = x_{5}$ state, we reach a target altitude. The control law is derived by following the Lyapunov methodology \cite{Khalil:1173048} applied to the sub-system (\ref{Newton's_Law (2D) finale1}). It is given by:

\begin{equation}\label{T_{1} control law}
T_{1} = \frac{M}{\cos(x_{5})} \pmb{\bigg[}    g - \dfrac{\beta\Theta_{2}}{M} + e_{3} + \ddot{x}_{3\ast} + k_{3}e_{4} - k^{2}_{3}e_{3} + k_{4}e_{4} \pmb{\bigg]},
\end{equation}
where $k_{3}$ and $k_{4}$ are control gains with positive value and the error terms are defined as follows:
\begin{equation}\label{error3}
e_{3} = x_{3\ast} - x_{3}
\end{equation}
\begin{equation}\label{error4}
e_{4} = x_{4\ast} - x_{4}.
\end{equation}
Now that we have a solution for $T_{1}$, we proceed by deriving a control law for the virtual control input $u_{y} = \sin(x_{5})$ responsible for the motion $y$, by using the same methodology. After some calculations the law becomes (see Appendix \ref{SecondAppendix})
\begin{equation}\label{L122}
u_{y} = \frac{M}{T_{1}} \pmb{\bigg[} - \dfrac{\beta\Theta_{1}}{M} +  e_{1} + \ddot{x}_{1\ast} + k_{5}e_{2} - k^{2}_{5}e_{1} + k_{6}e_{2}   \pmb{\bigg]},
\end{equation}
where $k_{5}$ and $k_{6}$ are control gains with positive values. In (\ref{L122}), the error terms are defined as follows:
\begin{equation}\label{L12}
e_{1} = x_{1\ast} - x_{1}
\end{equation}
\begin{equation}\label{L12}
e_{2} = x_{2\ast} - x_{2}.
\end{equation}
\begin{remark}
The function $u_{y}$ obtained in (\ref{L122}), is not a direct control input. However, it determines the desired pitch angle $x_{5\ast}$ which will be in turn used as the target pitch angle for the regulator of the swash masses displacement $\ell_{y}$ (see Fig.~\ref{Architecture}).
\end{remark}
Now, the desired pitch angle $x_{5\ast}=\phi_{\ast}$ which is the input to the pitch angle controller, is obtained as:
\begin{equation}\label{desired_roll}
x_{5\ast} = \sin^{-1}(u_{y}),
\end{equation}
where $u_{y}$ is given by (\ref{L122}).
Finally, the desired pitch angle obtained in the above equations serves as the input for the pitch angle controller.

\subsection{Pitch Angle Control}
Once we have set the desired pitch angle $x_{5\ast}=\phi_{\ast}$ according to (\ref{desired_roll}), the value of $\ell_{y}$ can be regulated via a control law derived by the application of the Lyapunov method as done in the previous sub-sections. The details can be found in the Appendix \ref{ThirdAppendix} and they lead to control law for $\ell_{y}$:
\begin{align}\label{e1l_1}
\ell_{y} {}&= \frac{I_{c}}{\beta T_{1}\cos(x_{5})} \pmb{\bigg[}  e_{5} + k_{1}e_{6} - k^{2}_{1}e_{5} + k_{2}e_{6}  \pmb{\bigg]},
\end{align}
where $k_{1}$ and $k_{2}$ are control gains with positive value.
\begin{remark}
The solution $\ell_{y}$ in (\ref{e1l_1}) is not necessarily constrained in the range $-L \leq \ell_{y} \leq L$.
\end{remark}
From the \textit{Remark 3}, in order to fulfill the physical constraint for $\ell_{y}$, we propose to introduce a saturation function to obtain
\begin{equation}\label{Saturation}
\bar{\ell}_{y} = sat(\ell_{y}) =
\begin{cases}
  -L & \hspace{1em} \ell_{y} < -L \\
  \ell_{y} & \hspace{1em} -L < \ell_{y} < L\\
  L &  \hspace{1em} \ell_{y} > L .
\end{cases}
\end{equation}
However, once the saturation occurs, the errors $e_{5}$ and $e_{6}$ may increase, which leads to an oscillation of the pitch angle. Therefore, the control law in (\ref{e1l_1}) can be bettered by designing a compensator and modifying the error signals that reduce the influence of the saturation. This leads to the following modified control law
\begin{align}\label{el_1new}
\bar\ell_{y,m}&= \frac{I_{c}}{\beta T_{1}\cos(x_{5})} \pmb{\bigg[} \bar{e}_{5} + k_{1}\bar{e}_{6} - k^{2}_{1}\bar{e}_{5} + k_{2}\bar{e}_{6}  \pmb{\bigg]},
\end{align}
where the new error definitions are
\begin{equation}\label{modified_errors1}
\bar{e}_{5} = e_{5} - e^{\star}
\end{equation}
\begin{equation}\label{modified_errors2}
\bar{e}_{6} = e_{6} - \dot{e^{\star}},
\end{equation}
and the auxiliary error is updated as follows:
\begin{equation}\label{update_auxiliary}
\dot{e^{\star}} = -\dfrac{\beta \epsilon_{1}}{I_{c}}e^{\star} + \dfrac{\beta(\bar{\ell}_{y,m} - \bar{\ell_{y}})}{I_{c}},
\end{equation}
with $\epsilon_{1}$ being a positive tuning parameter.
\begin{remark}
It should be observed that instead of saturating $\ell_{y}$, the target pitch angle from (\ref{desired_roll}) can be saturated, i.e., impose a maximum limit to $\phi$. However, we found that this provides worse performance.
\end{remark}

\subsection{Overall Control Algorithm}
The overall back-stepping control algorithm is sketched in Fig.~\ref{Architecture}. Firstly, the desired trajectories are generated. Feasible trajectories can be obtained for instance, with some waypoints that are located in the search space by the user \cite{1424025}, or they are generated by a stochastic approach \cite{aerospace4020027, ge}. Another approach is using dynamic path planning which the tangent vector field guidance (TVFG) and the Lyapunov vector field guidance (LVFG) are used \cite{6494384}. From an initial state, the control input $T_{1}$ is computed according to (\ref{T_{1} control law}) to reach a certain altitude. Then, the virtual control input $u_{y}$ is derived according to (\ref{L122}) to reach a certain horizontal coordinate. This induces a regulation of the pitch angle by setting a target value from (\ref{desired_roll}) and computing the swash masses position $\ell_{y}$ according to (\ref{el_1new}). The procedure is repeated iteratively at each time step to track the target trajectory. It should also be noted that the error terms are computed between the target values and the real UAV state given by the dynamical set of equations (\ref{Full_3D_model T})-(\ref{Full_3D_model R}). The time step (sampling period) is set to a given small value $T_s$.

\subsection{Extension to 3D}
In this section, we extend the results of the 2D case to derive the control laws for the full 3D system, following the same methodology. Essentially, the control mechanism comprises three steps. The control inputs in the 3D case can be grouped in the control input vector $\bm{u} = [T_{1}, \hspace{0.5em}\ell_{y}, \hspace{0.5em}\ell_{x}, \hspace{0.5em}M_{\psi}]^T$.

Observing (\ref{Full_3D_model R}), we can say that the rotational dynamics is highly coupled. To decouple the three second-order subsystems of the UAV in (\ref{Full_3D_model R}), we choose the globally defined change of input \cite{sastry1999nonlinear}
\begin{equation}\label{Decouple}
\begin{bmatrix}
  \beta T_{1}\ell_{y} \\
  \beta T_{1}\ell_{x} \\
  M_{\psi}
\end{bmatrix} = \begin{bmatrix}
                  A_{1} & A_{2} & A_{3} \\
                  A_{4} & A_{5} & A_{6} \\
                  A_{7} & A_{8} & A_{9}
                \end{bmatrix}\begin{bmatrix}
                               \upsilon_{1} \\
                               \upsilon_{2} \\
                               \upsilon_{3}
                             \end{bmatrix},
\end{equation}
where the elements $A_{i}$ are given in the Appendix \ref{ForthAppendix} and $\bm{\upsilon} = [\upsilon_{1}, \hspace{0.2em} \upsilon_{2}, \hspace{0.2em} \upsilon_{3}]^{T}$ is the new control input vector for the rotational dynamics in (\ref{Full_3D_model R}). This transformation yields the canonical form of the rotational dynamics of the UAV in (\ref{Full_3D_model R}).

Accordingly, our control laws for the 3D system will rely on a small displacement of the CM.

\textit{Step1}. \textit{Altitude control and horizontal position}: In this step, the control laws for $T_{1}$ (total thrust) and virtual control inputs $u_{x}$ and $u_{y}$ are formulated by using the same methodology of the 2D case. Following the Lyapunov methodology $T_{1}$ is obtained as follows:
\begin{equation}\label{T_{1} control law 3D}
T_{1} = \frac{M}{\cos(x_{5})\cos(x_{9})} \pmb{\bigg[}    g+ e_{3} + \ddot{x}_{3\ast} + k_{3}e_{4} - k^{2}_{3}e_{3} + k_{4}e_{4} \pmb{\bigg]},
\end{equation}
Now, we have the solution for $T_{1}$ in the 3D case, similarly, we proceed by deriving the control laws for $u_{x}$ and $u_{y}$ responsible for the motion $x-y$ motion, yielding
\begin{equation}\label{ux 3D}
u_{x} = \frac{M}{T_{1}} \pmb{\bigg[}  e_{7} + \ddot{x}_{7\ast} + k_{7}e_{8} - k^{2}_{7}e_{7} + k_{8}e_{8}   \pmb{\bigg]},
\end{equation}
where $k_{7}$ and $k_{8}$ are control gains. In (\ref{ux 3D}), the error terms are defined as follows:
\begin{equation}\label{e ux}
e_{7} = x_{7\ast} - x_{7} = x_{\ast} - x
\end{equation}
\begin{equation}\label{edot ux}
e_{8} = x_{8\ast} - x_{8} = \dot{x}_{\ast} - \dot{x} .
\end{equation}
\begin{equation}\label{uy 3D}
u_{y} = \frac{M}{T_{1}} \pmb{\bigg[}  e_{1} + \ddot{x}_{1\ast} + k_{5}e_{2} - k^{2}_{5}e_{1} + k_{6}e_{2}   \pmb{\bigg]},
\end{equation}
The virtual control laws (\ref{ux 3D}) and (\ref{uy 3D}) determine the target pitch and roll angles which will be in turn used as target pitch and roll angles for the regulator of $\ell_{y}$ and $\ell_{x}$. Now, the target pitch and roll angles are obtained from the virtual control inputs as:
\begin{equation}\label{dpitch 3D}
x_{5\ast} = \phi_{\ast} = \sin^{-1}\bigg(  \dfrac{u_{x}\sin(\psi_{\ast})-u_{y}\cos(\psi_{\ast})}{T_{1}}   \bigg)
\end{equation}
\begin{equation}\label{droll 3D}
x_{9\ast} = \theta_{\ast} = \sin^{-1}\bigg(  \dfrac{u_{x}\cos(\psi_{\ast})-u_{y}\sin(\psi_{\ast})}{T_{1}cos(x_{5\ast})}   \bigg).
\end{equation}
\textit{Step2}. \textit{Pitch and roll angle control}: Once we have set the target pitch $\phi_{\ast}$ and roll $\theta_{\ast}$ angles obtained through the virtual control laws $u_{x}$ and $u_{y}$, the values of $\ell_{y}$, and $\ell_{x}$ can be regulated. In detail,
\begin{equation}\label{e1l_1} \nonumber
\upsilon_{1} = I_{c} \pmb{\bigg[}   e_{5} + k_{1}e_{6} - k^{2}_{1}e_{5} + k_{2}e_{6}  \pmb{\bigg]},
\end{equation}
\begin{equation}\label{l1 3D} \nonumber
\upsilon_{2} = I_{cy} \pmb{\bigg[}  e_{9} + k_{9}e_{10} - k^{2}_{9}e_{9} + k_{10}e_{10}  \pmb{\bigg]},
\end{equation}
where $k_{9}$ and $k_{10}$ are control gains with positive value. $I_{yy} \approx I_{cy}$ is the approximated inertial term around the $y$ axis. The error terms in (\ref{l1 3D}) are defined as follows:
\begin{equation}\label{e ux}
e_{9} = x_{9\ast} - x_{9} = \theta_{\ast} - \theta
\end{equation}
\begin{equation}\label{edot ux}
e_{10} = x_{10\ast} - x_{10} = \dot{\theta}_{\ast} - \dot{\theta}.
\end{equation}
\textit{Step3}. \textit{Yaw control}: The desired yaw angle $\psi_{\ast}$ is then imposed so that the UAV heading and direction of motion follows the target path in the $x-y$ plane. From the illustration in Fig.~\ref{yaw} and trigonometry, the desired yaw angle $\psi_{\ast}$ can be obtained as follows:
      \begin{equation}\label{dyaw}
        \psi_{\ast} =  \tan^{-1}\bigg(\dfrac{y_{\ast}-y}{ x_{\ast}-x}\bigg).
      \end{equation}
\begin{proof}
The proof of the derived control laws using the Lyapunov approach can be obtained similarly to the 2D case. In particular, for Step 1 the derivation follows the same procedure in Appendix \ref{FirstAppendix} and \ref{SecondAppendix}, while for Step 2 the one in Appendix \ref{ThirdAppendix}. Then, the control law in Step 3 follows straightforwardly.
\end{proof}

\begin{table}[h]
\renewcommand{\arraystretch}{1.3}
\caption{{Back-stepping controller tuning parameters} }
\label{Table-Ig}
\centering
\begin{tabular}{|c |c| c| c| c| c| c|c|}
\hline
\bfseries Parameter  & $k_{1}$ &  $k_{2}$ & $k_{3}$ &  $k_{4}$ &  $k_{5}$ &  $k_{6}$ & $\epsilon_{1}$   \\\hline
\bfseries Linear trajectory     & $0.2$ & $3$ & $0.2$ & $2$ & $0.2$ & $2$ & $0.1$ \\\hline
\bfseries Complex trajectory     & $5$ & $0.5$ & $1$ & $2$ & $1.6$ & $8$ & $0.2$ \\
\hline
\end{tabular}
\end{table}

\begin{center}
\begin{figure}[h]
\includegraphics[scale=0.6]{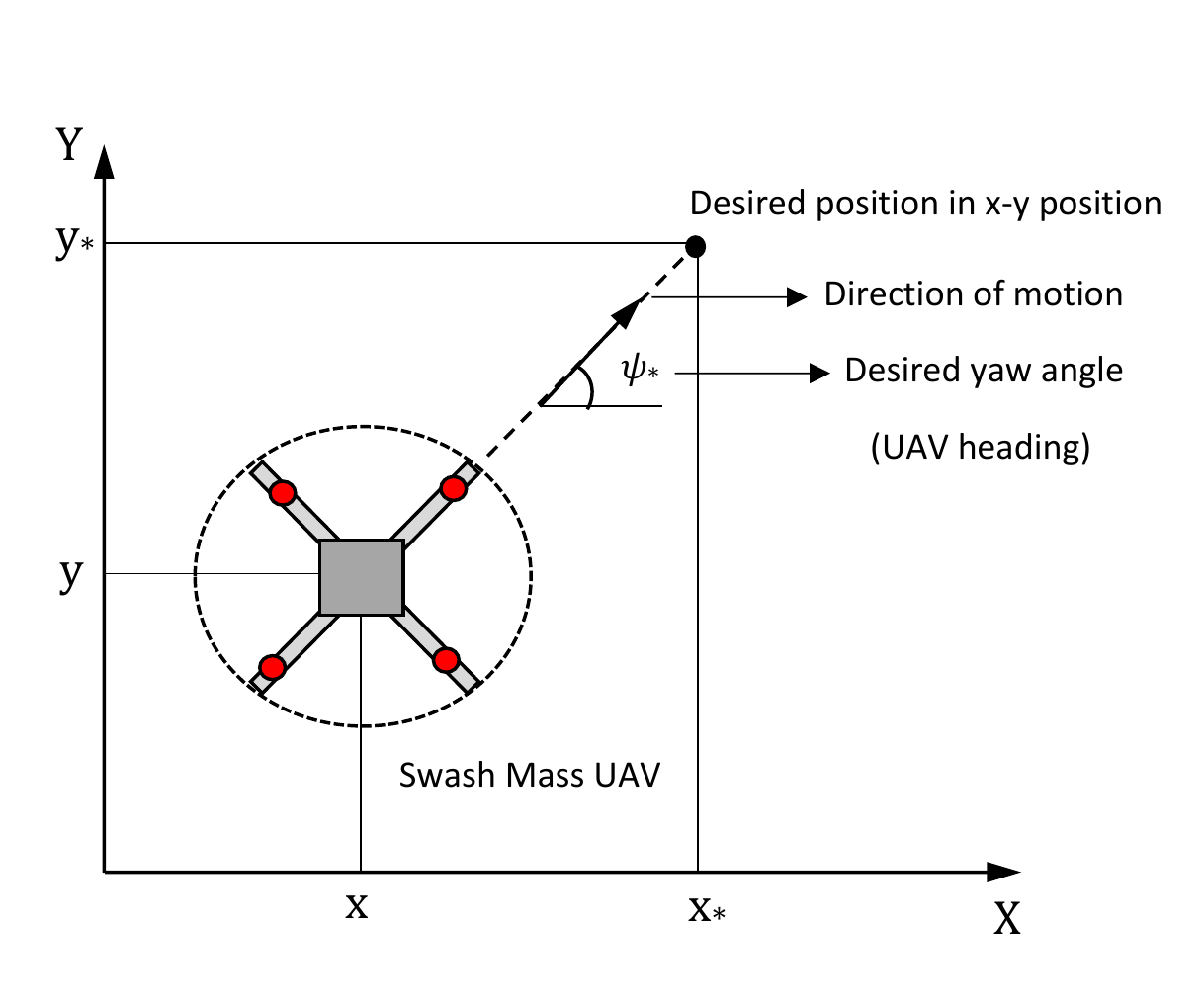}
\centering
\vspace{-0.1em}
\caption{Top view of the swash mass unmanned aerial vehicle structure, motion in x-y plane.}
\vspace{-0.5em}
\label{yaw}
\centering
\end{figure}
\end{center}

\section{numerical results}
In this section, numerical results are reported to evaluate the proposed control strategy for tracking given geometric trajectories to be followed by the UAV. Two flying scenarios are considered: a linear trajectory, and a complex trajectory. The total time of flight for both scenarios are set to $T_{f} = 10 s$ and $T_{f} = 14 s$, respectively, the sampling time is set to $T_{s} = 0.1 ms$, and the initial conditions are $\bm{x}_{init} = [0, 0, 0, 0, 0, 0]^{T}$. For simplicity of exposition, the 2D planner case is shown. The value of the controller tuning parameters are reported in Table \ref{Table-Ig}.

\subsection{Linear Trajectory}

We start by considering a linear trajectory. In detail, this is given by
\begin{align}\label{trajectory1}
y_{\ast}(t) = x_{1\ast}(t) = 0.857t, \\
z_{\ast}(t) = x_{3\ast}(t) = 0.857t,
\end{align}

\begin{figure}[t]
\includegraphics[scale=0.3]{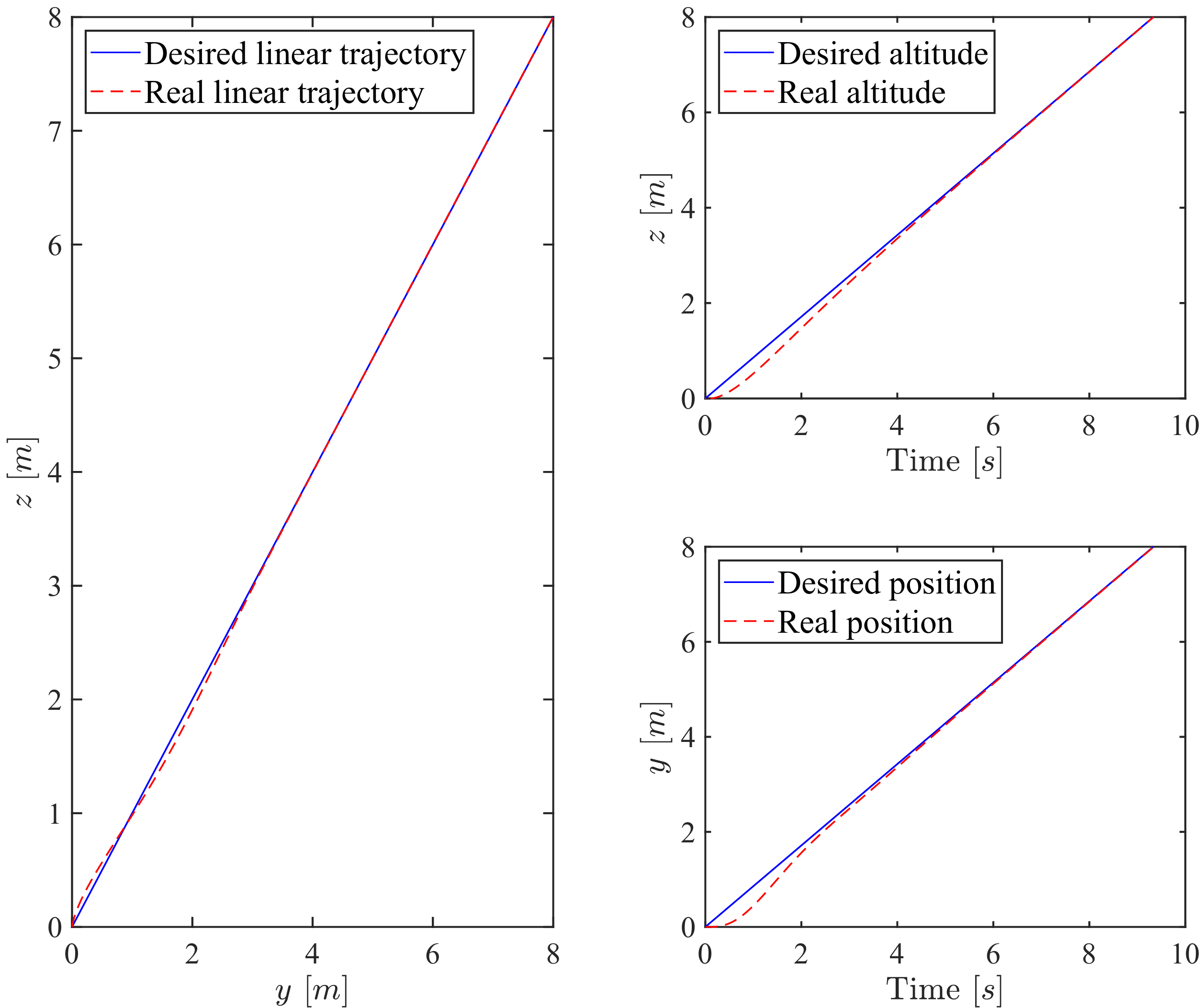}
\centering
\vspace{-0.1em}
\caption{Linear maneuver simulation.}
\vspace{-1.0em}
\label{Linear_trajectory}
\end{figure}
\begin{figure}[t]
\includegraphics[scale=0.3]{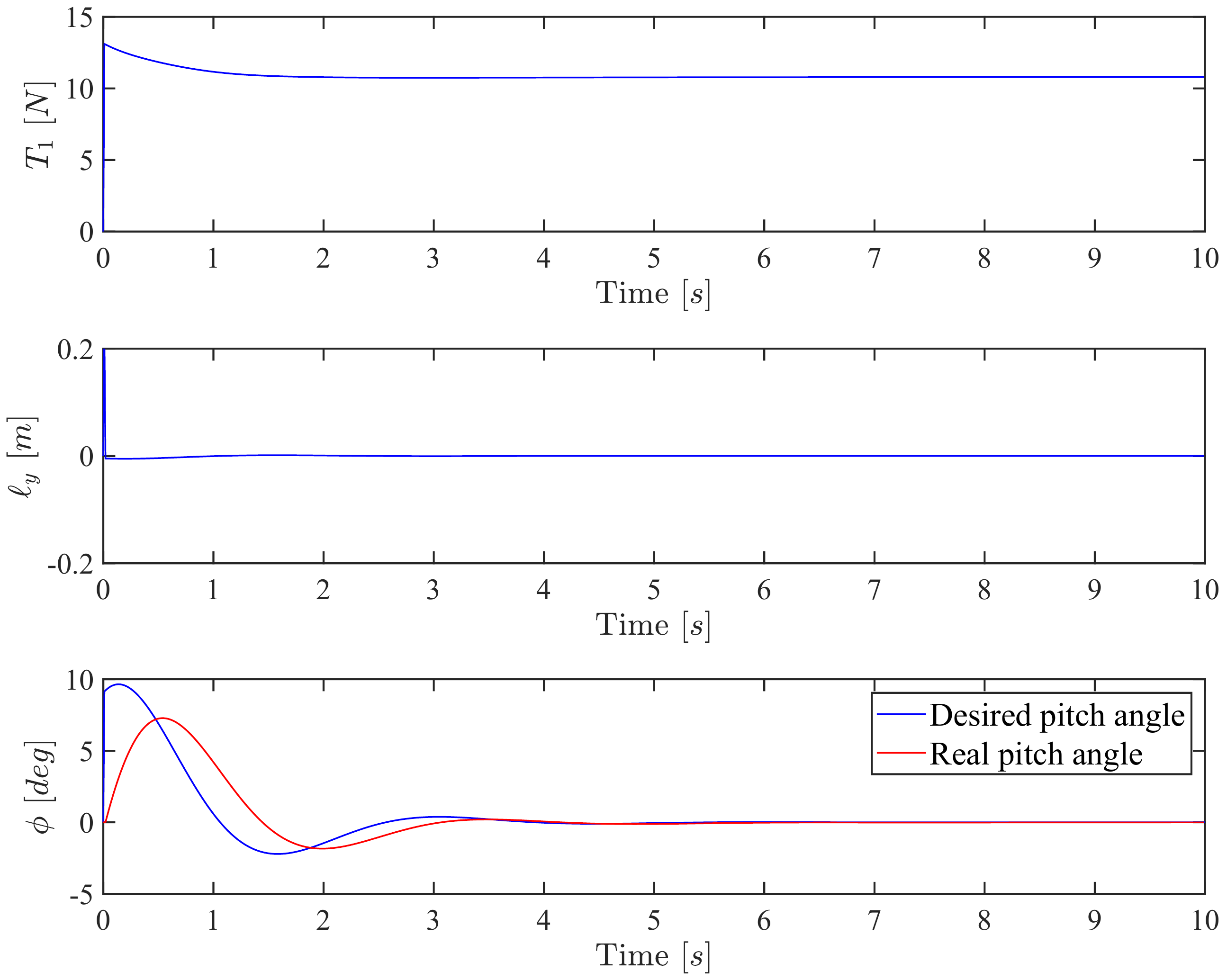}
\centering
\vspace{-0.1em}
\caption{Control inputs $\bm{u} = [T_{1},\hspace{0.5em} \ell_{y}]^{T}$ and pitch angle $\phi$ for the linear maneuver. }
\vspace{-1.0em}
\label{Linear_control_inputs}
\end{figure}

In Fig.~\ref{Linear_trajectory}, both the target and the real trajectories are shown and the difference is not pronounced. Quantitatively, in Table.~\ref{Table-rmse}, we report the root-mean-square-error (RMSE) between the desired trajectory components and the real trajectory components. The overall RMSE is equal to $0.3 m$.

The control inputs $\bm{u} = [T_{1},\hspace{0.5em} \ell_{y}]^{T}$ are reported in Fig.~\ref{Linear_control_inputs}. The target pitch angle obtained through the inverse of the virtual control input $u_y$ is shown in the third sub-plot of Fig.~\ref{Linear_control_inputs}. The actual (real) pitch angle of the UAV is also shown herein. The pitch angle oscillation is more pronounced at the beginning of the flight and then it converges to a constant value in steady state.

\subsection{Complex Trajectory}

We now consider a complex trajectory given by the relations
\begin{align}\label{trajectory2}
y_{\ast}(t) = x_{1\ast}(t) = 4\sin(0.5t), \\
z_{\ast}(t) = x_{3\ast}(t) = 5\sin(t).
\end{align}

\begin{figure}[t]
\includegraphics[scale=0.3]{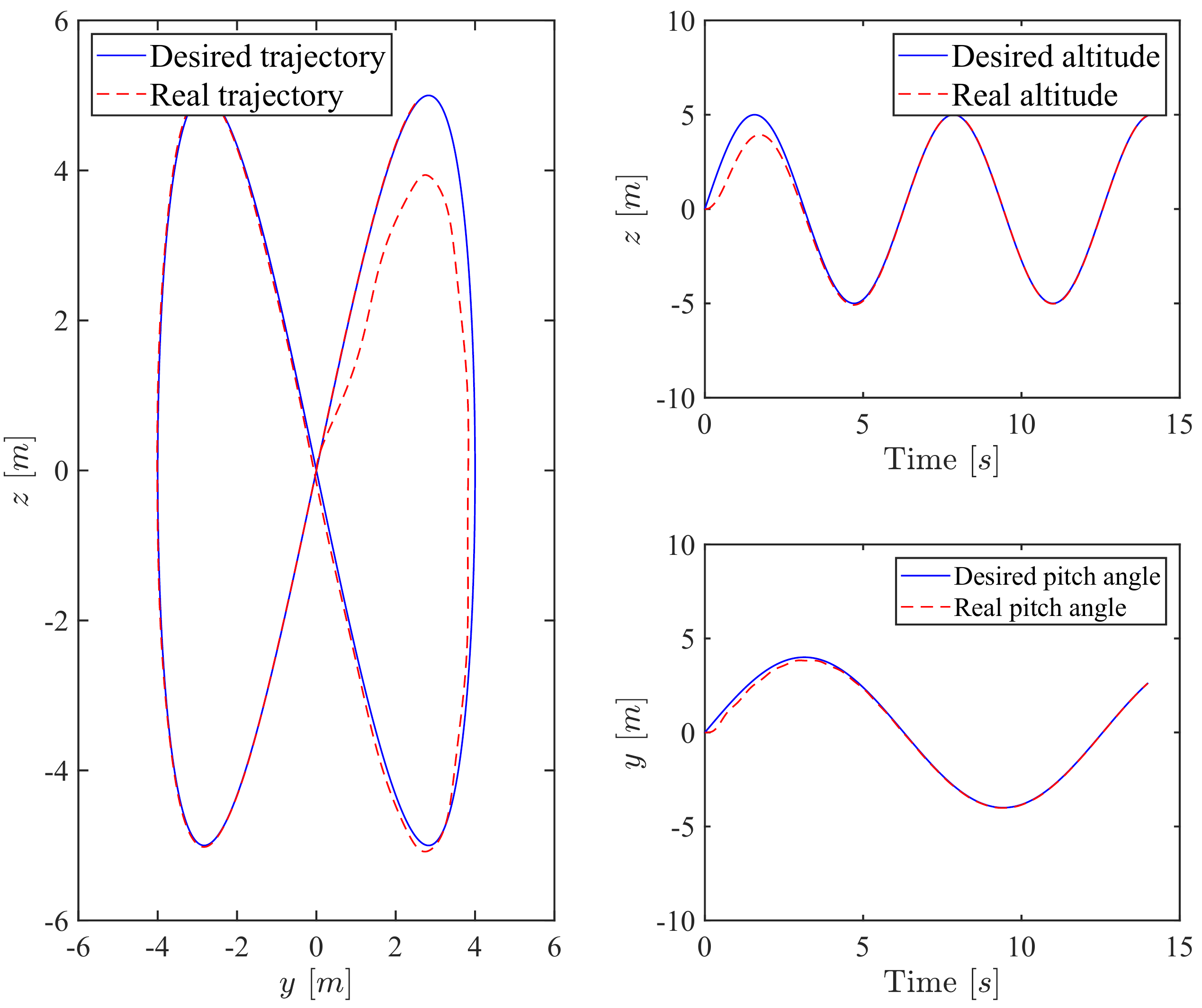}
\centering
\vspace{-0.1em}
\caption{Complex maneuver simulation.}
\vspace{-1.0em}
\label{Complex_trajectory}
\end{figure}
\begin{figure}[t]
\includegraphics[scale=0.3]{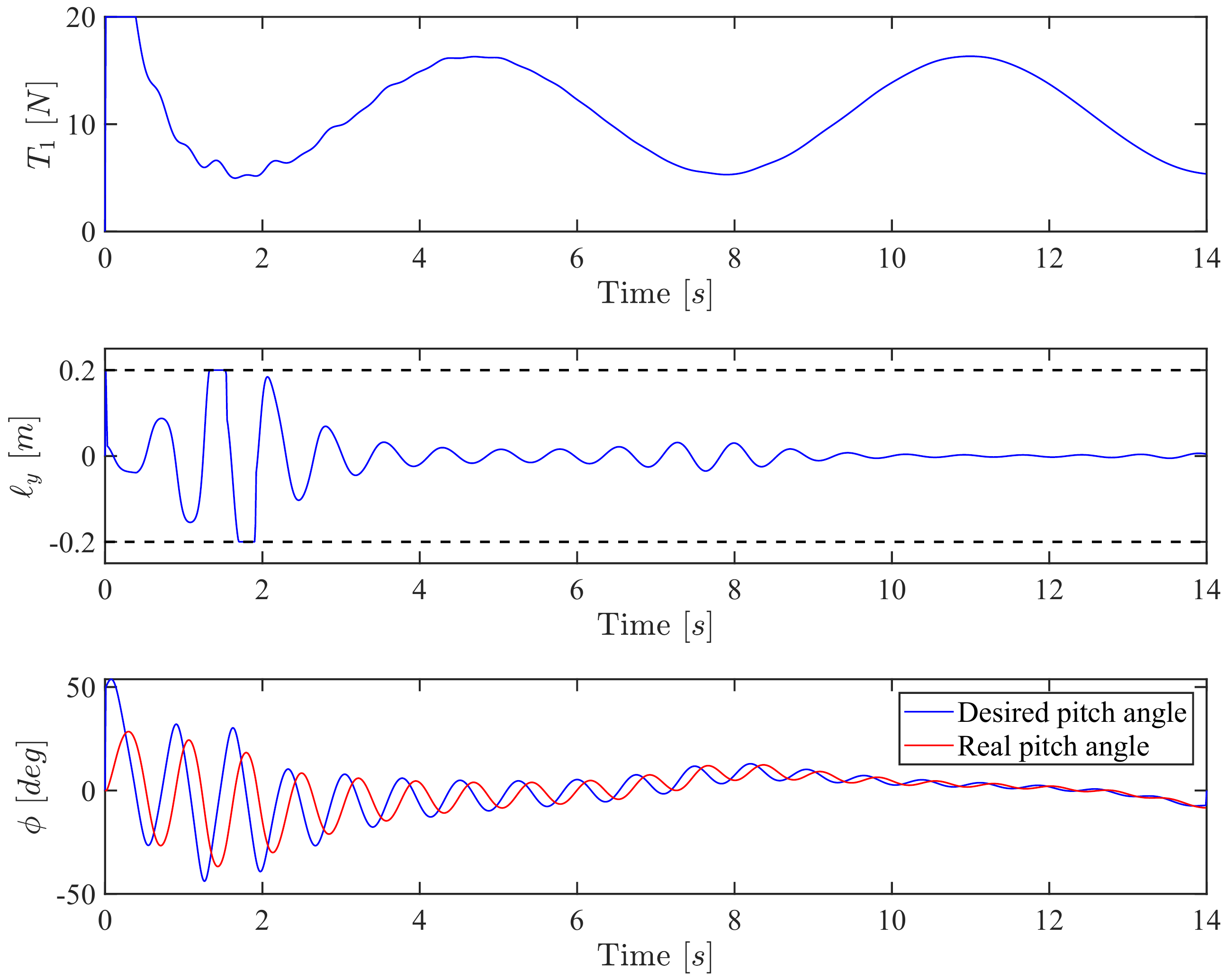}
\centering
\vspace{-0.1em}
\caption{Control inputs $\bm{u} = [T_{1},\hspace{0.5em} \ell_{y}]^{T}$ and pitch angle $\phi$ for the complex trajectory. }
\vspace{-1.0em}
\label{Complex_control_inputs}
\end{figure}

It simulates the behavior of the UAV in an aggressive maneuver, from a given orientation, upside down for the pitch angle $\phi$, and a strong lateral motion $y$ and altitude motion $z$.
The results are reported in Fig.~\ref{Complex_trajectory}. Despite the aggressiveness of the trajectory, the UAV is able to well follow it. The real-target trajectory RMSE is reported in Table~\ref{Table-rmse} and the overall RMSE equals $0.35 m$.

\begin{table}[h]
\renewcommand{\arraystretch}{1.3}
\caption{{Root-mean-square error between desired and real trajectory} }
\label{Table-rmse}
\centering
\begin{tabular}{|c|c|c|}
\hline
\bfseries Mission type  & \bfseries Value (m)   \\\hlineB{2.5}
\hline \hline
$RMSE_{y}$ for linear trajectory    & $0.2979$   \\\hline
$RMSE_{z}$ for linear trajectory    & $0.3102$\\\hline
$RMSE_{y}$ for complex trajectory  & $0.1507$ \\\hline
$RMSE_{z}$ for complex trajectory  & $ 0.5589$\\
\hline
\end{tabular}
\end{table}

The control inputs for the complex trajectory are shown in Fig.~\ref{Complex_control_inputs}. It is interesting to note that in the time window between $1 s$ and $3 s$ the control input $\ell_{y}$ is saturated to the value $\bar{\ell}_{y} = L = 0.2 m$. Furthermore, the evolution of the control inputs follows a dumped sinusoidal shape, and this is directly related to the target motion trajectory in the $y-z$ plane.

\section{conclusion}
We have presented a new unmanned aerial vehicle (UAV) structure that allows maneuvering the UAV through the control of four swash masses and the double blade rotor thrust. A dynamical system model has been derived from the Newton's laws. It has been shown that a steady state can be defined and it corresponds to the state where the swash mass UAV is at rotation equilibrium and follows a linear trajectory. We have then focused the attention to the design of an automatic control mechanism so that the UAV follows a certain target trajectory. The dynamical system of equations that describes the UAV dynamics is non-linear and rather unique. In fact, it consists of sub-sets of differential equations coupled through the control inputs. Furthermore, the inertia term is time dependent as a result of the swash masses movement, which introduces a dependency on both the controllable masses positions and theirs derivative.
A back-stepping control approach has then be derived considering physical constraints. Several results from simulations have been presented to assess the performance of the UAV that is maneuvered to follow both linear and aggressive trajectories. They show that the UAV can be controlled and such trajectories can be well followed with small RMSE.

\appendices
\section{Derivation of the Control Law for $T_{1}$}
\label{FirstAppendix}
To derive the control law for $T_{1}$, so that a target altitude is reached we start from equation (\ref{T_{1} control law}) and follow the Lyapunov methodology \cite{Khalil:1173048} and (\cite{zhou2008adaptive}, Chapter 2), to identify a suitable Lyapunov candidate function.
Firstly, the error term for the altitude $x_{3}=z$ is defined
\begin{equation}\label{error-z}
e_{3} = x_{3\ast}-x_{3}.
\end{equation}
A suitable candidate Lyapunov function is chosen as
\begin{equation}\label{L3}
V(e_{3}) = \frac{1}{2}\bigg(e^{2}_{3}\bigg).
\end{equation}
The first derivative of $V(e_{3})$ with respect to time is given by
\begin{equation}\label{L3}
\dot{V}(e_{3}) = e_{3}\dot{e_{3}} = e_{3}(\dot{x}_{3\ast}-\dot{x}_{3}).
\end{equation}
If $x_{4} = \dot{x}_{3\ast} + k_{3}e_{3}$, for $k_{3}>0$, then $ \dot{V}(e_{3})$ is negative semi-definite and the error term $e_{3}$ converges to zero. Thus, $x_{4\ast}$ is defined as
\begin{equation}\label{L33}
x_{4\ast} = \dot{x}_{3\ast} + k_{3}e_{3}, \hspace{1em} for \hspace{0.5em}k_{3} > 0.
\end{equation}
To get $x_{4\ast}$ in (\ref{L33}), we should be able to control $x_{4}$ which comes from the dynamics equation defined in (\ref{state-space-dynamic3}). Therefore, another error term $e_{4}$ can be defined
\begin{equation}\label{L3}
e_{4} = x_{4\ast} - x_{4} = \dot{x}_{3\ast} + k_{3}e_{3} - x_{4}.
\end{equation}
Since we want both the error terms $e_{3}$ and $e_{4}$ to converge to zero, an augmented Lyapunov function of $e_{3}$ and $e_{4}$ is chosen:
\begin{equation}\label{L34}
V(e_{3}, e_{4}) = \frac{1}{2}\bigg(e^{2}_{3} + e^{2}_{4}\bigg).
\end{equation}
By differentiating of $V(e_{3}, e_{4})$ with respect to time, we obtain
\begin{equation}\label{L344}
\dot{V}(e_{3}, e_{4}) = e_{3}\dot{e}_{3} + e_{4}\dot{e}_{4}.
\end{equation}
Now, we need $\dot{e}_{3}$ and $\dot{e}_{4}$ so that they can be replaced into (\ref{L344}). If we differentiate (\ref{error-z}) and considering (\ref{state-space-dynamic}), $\dot{e}_{3}$ is obtained:
\begin{equation}\label{L34}
\dot{e}_{3} = \dot{x}_{3\ast} - \dot{x}_{3} = \dot{x}_{3\ast} - x_{4}.
\end{equation}
Then
\begin{equation}\label{L341}
\dot{e}_{3} = e_{4} - k_{3}e_{3}.
\end{equation}
Furthermore, by differentiating (\ref{L3}), $\dot{e}_{4}$ is obtained
\begin{equation}\label{L345}
\dot{e}_{4} = \ddot{x}_{3\ast} + k_{3}\dot{e}_{3} - \dot{x}_{4}.
\end{equation}
By replacing (\ref{L341}) and (\ref{L345}) into (\ref{L344}), $\dot{V}(e_{3}, e_{4})$ results as
\begin{equation}\label{L3444}
\dot{V}(e_{3}, e_{4}) = e_{3}e_{4} - k_{3}e^{2}_{3} + e_{4}\ddot{x}_{3\ast} + k_{3}\dot{e}_{3}e_{4} - e_{4}\dot{x}_{4}.
\end{equation}
In (\ref{L3444}), $\dot{x}_{4}$ comes from the dynamical system defined in (\ref{state-space-dynamic2}) and (\ref{state-space-dynamic3}) and it can be replaced by the term $-g + \frac{\beta\Theta_{2}}{M} + \frac{T_{1}\cos(x_{5})}{M}$. The final form of $\dot{V}(e_{3}, e_{4})$ is obtained as
\begin{align}\label{L34444} \nonumber
  \dot{V}(e_{3}, e_{4}) = e_{3}e_{4} - k_{3}e^{2}_{3} + e_{4}\ddot{x}_{3\ast} + k_{3}e^{2}_{4}  \\
  - e_{4} \pmb{\bigg(} -g + \dfrac{\beta\Theta_{2}}{M} + \dfrac{T_{1}\cos(x_{5})}{M} \pmb{\bigg)}.
\end{align}

In (\ref{L34444}), $T_{1}$ represents the control force. By choosing $T_{1}$ properly, $\dot{V}(e_{3}, e_{4})$ becomes negative semi-definite and the error terms $e_{3}$ and $e_{4}$ converge to zero. Thus, a suitable control law for $T_{1}$ is chosen as
\begin{equation}\label{L3455}
T_{1} = \dfrac{M}{\cos(x_{5})} \pmb{\bigg(}    g - \dfrac{\beta\Theta_{2}}{M} + e_{3} + \ddot{x}_{3\ast} + k_{3}e_{4} - k^{2}_{3}e_{3} + k_{4}e_{4} \pmb{\bigg)},
\end{equation}
where $k_{3}$ and $k_{4}$ are control gains with positive value. In fact, by substituting (52) in (51), $\dot{V}(e_{3}, e_{4})$ becomes
\begin{equation}\label{L34}
\dot{V}(e_{3}, e_{4}) = -k_{3}e^{2}_{3} - k_{4}e^{2}_{4} < 0 ,  \hspace{1em} for \hspace{0.5em}k_{3} > 0 , \hspace{0.5em} k_{4} > 0.
\end{equation}

\begin{remark}
In (\ref{L3455}), the second order derivative of the desired altitude $\ddot{x}_{3\ast}$ is added to the control law to increase the tracking performance.
\end{remark}

\section{Derivation of the control law for $u_{y}$}
\label{SecondAppendix}
In this Appendix the virtual control input $u_{y}$ is formulated to obtain a desired pitch angle $x_{5\ast} = \phi_{\ast}$. Following the Lyapunov methodology, let us consider the error term
\begin{equation}\label{error-y}
e_{1} = x_{1\ast}-x_{1}.
\end{equation}
and consequently define the Lyapunov function
\begin{equation}\label{Ly}
V(e_{1}) = \frac{1}{2}\bigg(e^{2}_{1}\bigg).
\end{equation}
We now compute its time derivative to obtain
\begin{equation}\label{Ly1}
\dot{V}(e_{1}) = e_{1}\dot{e_{1}} = e_{1}(\dot{x}_{1\ast}-\dot{x}_{1}).
\end{equation}
If $x_{2} = \dot{x}_{1\ast} + k_{5}e_{1}$ for $k_{5}>0$, then $ \dot{V}(e_{1})$ is negative semi-definite and the error term $e_{1}$ converges to zero. To accomplish this we introduce
\begin{equation}\label{Ly2}
x_{2\ast} = \dot{x}_{1\ast} + k_{5}e_{1}, \hspace{1em} for \hspace{0.5em}k_{5} > 0.
\end{equation}
To get $x_{2\ast}$ in (\ref{Ly2}), we should be able to control $x_{2}$ which comes from the dynamics given by (\ref{state-space-dynamic1}). Therefore, another error term $e_{2}$ can be defined
\begin{equation}\label{Ly3}
e_{2} = x_{2\ast} - x_{2} = \dot{x}_{1\ast} + k_{1}e_{1} - x_{2}.
\end{equation}
Now, we want both the error terms $e_{1}$ and $e_{2}$ to converge to zero. Therefore, an augmented Lyapunov function of $e_{1}$ and $e_{2}$ is chosen
\begin{equation}\label{Ly4}
V(e_{1}, e_{2}) = \frac{1}{2}\bigg(e^{2}_{1} + e^{2}_{2}\bigg).
\end{equation}
By differentiating $V(e_{1}, e_{2})$ with respect to the time, we obtain
\begin{equation}\label{Ly5}
\dot{V}(e_{1}, e_{2}) = e_{1}\dot{e}_{1} + e_{2}\dot{e}_{2}.
\end{equation}
After straightforward algebraic manipulations, we get
\begin{equation}\label{Ly6}
\dot{e}_{1} = \dot{x}_{1\ast} - \dot{x}_{1} = \dot{x}_{1\ast} - x_{2},
\end{equation}
then
\begin{equation}\label{Ly7}
\dot{e}_{1} = e_{2} - k_{5}e_{1}.
\end{equation}
Furthermore, by differentiating (\ref{Ly3}) with respect to time, $\dot{e}_{2}$ becomes
\begin{equation}\label{Ly8}
\dot{e}_{2} = \ddot{x}_{1\ast} + k_{5}\dot{e}_{1} - \dot{x}_{2}.
\end{equation}
By replacing (\ref{Ly7}) and (\ref{Ly8}) into (\ref{Ly5}), $\dot{V}(e_{1}, e_{2})$ is obtained as follows
\begin{equation}\label{Ly9}
\dot{V}(e_{1}, e_{2}) = e_{1}e_{2} - k_{5}e^{2}_{1} + e_{2}\ddot{x}_{1\ast} + k_{5}\dot{e}_{1}e_{2} - e_{2}\dot{x}_{2}.
\end{equation}
In (\ref{Ly9}), $\dot{x}_{2}$ comes from the dynamical system of equations defined in (\ref{state-space-dynamic}) and (\ref{state-space-dynamic1}). It can be replaced by the term $ \frac{\beta\Theta_{1}}{M} + \frac{T_{1}\sin(x_{5})}{M}$. The final form of $\dot{V}(e_{1}, e_{2})$ is obtained as
\begin{align}\label{Ly10} \nonumber
  \dot{V}(e_{1}, e_{2}) = e_{1}e_{2} - k_{5}e^{2}_{1} + e_{2}\ddot{x}_{1\ast} + k_{5}e^{2}_{2}  \\
  - e_{2} \pmb{\bigg(} \dfrac{\beta\Theta_{1}}{M} + \dfrac{T_{1}\sin(x_{5})}{M} \pmb{\bigg)}.
\end{align}
In (\ref{Ly10}), $u_{y} = \sin(x_{5})$ represents the virtual control input. By choosing $u_{y}$ properly, $\dot{V}(e_{1}, e_{2})$ becomes negative semi-definite so that the error terms $e_{1}$ and $e_{2}$ converge to zero. Thus, a suitable control law for $u_{y}$ is chosen as
\begin{equation}\label{Ly11}
u_{y} = \dfrac{M}{T_{1}} \pmb{\bigg(} -\dfrac{\beta\Theta_{1}}{M} +  e_{1} + \ddot{x}_{1\ast} + k_{5}e_{2} - k^{2}_{5}e_{1} + k_{6}e_{2}   \pmb{\bigg)},
\end{equation}
where $k_{5}$ and $k_{6}$ are control gains with positive value. This is because, by substituting (\ref{Ly11}) in (\ref{Ly10}), $\dot{V}(e_{1}, e_{2})$ becomes negative:
\begin{equation}\label{Ly12}
\dot{V}(e_{1}, e_{2}) = -k_{1}e^{2}_{1} - k_{5}e^{2}_{2} < 0 ,  \hspace{1em} for \hspace{0.5em}k_{5} > 0 , \hspace{0.5em} k_{6} > 0.
\end{equation}

\section{Derivation of the control law for $\ell_{y}$}
\label{ThirdAppendix}
The control law for $\ell_{y}$ is formulated in this Appendix following a similar methodology to the one used to derive the control laws for $T_1$ and $u_y$. Motivated by (\ref{state-space-dynamic4}) and (\ref{state-space-dynamic5}), the pitch tracking error is expressed as
\begin{equation}\label{err_pitch}
e_{5} = x_{5\ast} - x_{5},
\end{equation}
where $ x_{5\ast}$ represents the desired value for the pitch angle $\phi$. We consider the Lyapunov candidate function
\begin{equation}\label{L51}
V(e_{5}) = \frac{1}{2}\bigg(e^2_{5}\bigg).
\end{equation}
Its time derivative is
\begin{equation}\label{L52}
\dot{V}(e_{5}) = e_{5}\dot{e_{5}} = e_{5}(\dot{x}_{5\ast}-\dot{x}_{5}).
\end{equation}
By using the dynamical system equations in (\ref{state-space-dynamic4}) and (\ref{state-space-dynamic5}), we can replace $\dot{x}_{5}$ by $x_{6}$ and (\ref{L52}) becomes
\begin{equation}\label{L53}
\dot{V}(e_{5}) = e_{5}(\dot{x}_{5\ast}-x_{6}).
\end{equation}
The stabilization of $e_{5}$ in (\ref{err_pitch}) can be satisfied by introducing $x_{6} = \dot{x}_{5\ast} + k_{1}e_{5}$. Then $\dot{V}(e_{5})$ is negative semi-definite for $k_{1}>0$. Because the fully actuated subsystem in (\ref{state-space-dynamic4}) and (\ref{state-space-dynamic5}) is a second order system we also want to stabilize the derivative of the pitch angle. Therefore, $x_{6\ast} =  \dot{x}_{5\ast} + k_{1}e_{5}$ for $k_{1}>0$. To stabilize the derivative of the pitch angle we need another error term which converges to zero. This error term is defined as follows
\begin{equation}\label{L54}
e_{6} = x_{6\ast} - x_{6}.
\end{equation}
To track $x_{5\ast}$ and $x_{6\ast}$, the error terms $e_{5}$ and $e_{6}$ should converge to zero. We consider therefore the augmented Lyapunov function
\begin{equation}\label{L55}
V(e_{5},e_{6}) = \frac{1}{2}\bigg(e^{2}_{5} + e^{2}_{6}\bigg).
\end{equation}
The time derivative of (\ref{L55}) is
\begin{equation}\label{L56}
\dot{V}(e_{5},e_{6}) = e_{5}\dot{e}_{5} + e_{6}\dot{e_{6}}.
\end{equation}
By taking (\ref{L55}) and $x_{6} = \dot{x}_{5\ast} + k_{1}e_{5}$ into account, $e_{6}$ is written as
\begin{equation}\label{L57}
e_{6} = \dot{x}_{5\ast} + k_{1}e_{5} - x_{6}.
\end{equation}
By replacing $\dot{e}_{5}$ with the term $(\dot{x}_{5\ast} - x_{6})$ in (\ref{L57}), then $e_{6}$ becomes
\begin{equation}\label{LL58}
e_{6} = \dot{e}_{5} + k_{1}e_{5}.
\end{equation}
By differentiating (\ref{LL58}) with respect to time, $\dot{e}_{6}$ is obtained as
\begin{equation}\label{LL59}
\dot{e}_{6} = \ddot{x}_{5\ast} - \dot{x}_{6} + k_{1}\dot{e}_{5}.
\end{equation}
By replacing (\ref{LL59}) into (\ref{L56}), $\dot{V}(e_{5},e_{6})$ is given by
\begin{equation}\label{LL510}
\dot{V}(e_{5},e_{6}) = e_{5}(e_{6}-k_{1}e_{5}) + e_{6}(\ddot{x}_{5\ast} + k_{1}\dot{e}_{5}-\dot{x}_{6}).
\end{equation}
In (\ref{LL510}), $\dot{x}_{6}$ comes from the dynamical equations (\ref{state-space-dynamic4}) and (\ref{state-space-dynamic5}). By replacing $\dot{x}_{6}$ with $ \frac{\beta T_{1}\cos(x_{5})\ell_{y}}{I_{c}}$ and with some simple calculations, finally we get
\begin{align}
\begin{split}\nonumber
\dot{V}(e_{5},e_{6}) = e_{5}e_{6} - k_{1}e^{2}_{5} + \ddot{x}_{5\ast}e_{6} + k_{1}e^2_{6}
\end{split}\\
\begin{split}\label{LL511}
 - k^{2}_{1}e_{5}e_{6} - e_{6}\pmb{\bigg(}  \dfrac{\beta T_{1}\cos(x_{5})\ell_{y}}{I_{c}}  \pmb{\bigg)} .
 \end{split}
\end{align}
In (\ref{LL511}), $\ell_{y}$ represents the control input. Now we should choose $\ell_{y}$ such that $\dot{V}(e_{5},e_{6})$ is negative semi-definite. Therefore, the control law of $\ell_{y}$ for the roll angle is chosen as
\begin{align}\label{el_1}
\ell_{y} {}&= \dfrac{I_{c}}{\beta T_{1}\cos(x_{5})} \pmb{\bigg(}  e_{5} + k_{1}e_{6} - k^{2}_{1}e_{5} + k_{2}e_{6}  \pmb{\bigg)},
\end{align}
where $k_{1}$ and $k_{2}$ are control gains with positive value. By substituting (\ref{el_1}) into (\ref{LL511}), $\dot{V}(e_{5},e_{6})$ becomes:
\begin{equation}\label{L511}
 \dot{V}(e_{5},e_{6}) = -k_{1}e^{2}_{5} - k_{2}e^{2}_{6}  \hspace{1em} for \hspace{0.5em}k_{1} > 0 , \hspace{0.5em} k_{2} > 0
\end{equation}
We can conclude that if the control law of $\ell_{y}$ is chosen as (\ref{el_1}), $\dot{V}(e_{5},e_{6})$ is negative semi-definite and the convergence of the error terms $e_{5}$ and $e_{6}$ to zero is fulfilled.

\begin{remark}
It should be observed that the errors $\bar{e}_{5}$ and $\bar{e}_{6}$ introduced in (\ref{modified_errors1}) and (\ref{modified_errors2}) converge asymptotically to zero since if we consider the Lyapunov function
\begin{equation}\label{L55new}
V(\bar{e}_{5},\bar{e}_{6}) = \frac{1}{2}\bigg(\bar{e}^{2}_{5} + \bar{e}^{2}_{6}\bigg),
\end{equation}
by substituting (\ref{update_auxiliary}) and (\ref{el_1new}) to its derivative,
\begin{equation}\label{L511}
 \dot{V}(\bar{e}_{5},\bar{e}_{6}) = -k_{1}\bar{e}^{2}_{5} - k_{2}\bar{e}^{2}_{6},
\end{equation}
this is negative for all positive parameters $k_{1}$ and $k_{2}$. Consequently, since also the auxiliary error goes to zero, the pitch angle error is also asymptotically null.
\end{remark}

\section{The decoupling matrix}
\label{ForthAppendix}
The elements $A_{i}, i = 1, ..., 9$, in (\ref{Decouple}) are
\begin{align}
\begin{split}\nonumber
A_{1} = \cos(\phi)\sin(\psi) - \cos(\psi)\sin(\phi)\sin(\theta)
\end{split}\\
\begin{split}\nonumber
A_{2} = \cos(\psi)\cos(\theta)
\end{split}\\
\begin{split}\nonumber
A_{3} = \sin(\phi)\sin(\psi) + \cos(\phi)\cos(\psi)\sin(\theta)
\end{split}\\
\begin{split}\nonumber
A_{4} = \cos(\phi)\cos(\psi) + \sin(\phi)\sin(\psi)\cos(\theta)
\end{split}\\
\begin{split}\nonumber
A_{5} = \cos(\theta)\sin(\psi)
\end{split}\\
\begin{split}\nonumber
A_{6} = -\cos(\psi)\sin(\phi) + \cos(\phi)\sin(\psi)\sin(\theta)
\end{split}\\
\begin{split}\nonumber
A_{7} = \cos(\theta)\sin(\phi)
\end{split}\\
\begin{split}\nonumber
A_{8} = \sin(\theta)
\end{split}\\
\begin{split}\nonumber
A_{9} = \cos(\phi)\cos(\theta).
\end{split}
\end{align}

\ifCLASSOPTIONcaptionsoff
  \newpage
\fi



%
\bibliographystyle{IEEEtran}
\bibliography{bib}
\end{document}